\documentclass[11pt,letterpaper]{article}
\usepackage[margin=1in]{geometry}
\usepackage[T1]{fontenc}
\usepackage{amsmath,amsthm,mathtools}
\usepackage{lmodern,amssymb}
\usepackage{microtype}
\usepackage{nicefrac}
\usepackage{enumitem,booktabs,array,multirow}
\usepackage{tabularx,graphicx}
\usepackage{needspace}
\usepackage[table,dvipsnames]{xcolor}
\definecolor{linkblue}{HTML}{254B78}
\definecolor{accentblue}{HTML}{315F7D}
\definecolor{accentlight}{HTML}{EAF1F5}
\definecolor{slate}{HTML}{5D6872}
\definecolor{softgray}{HTML}{F2F4F5}
\usepackage{tikz}
\usetikzlibrary{arrows.meta,calc,decorations.pathreplacing,positioning}
\usepackage[square,numbers]{natbib}
\definecolor{citepurple}{HTML}{2A2AA5}
\usepackage[colorlinks=true,linkcolor=linkblue,citecolor=citepurple,urlcolor=linkblue,bookmarksnumbered=true]{hyperref}
\usepackage{aliascnt}
\usepackage[nameinlink,capitalise,noabbrev]{cleveref}
\AtBeginDocument{%
  \setlength{\abovedisplayskip}{8pt plus 2pt minus 2pt}%
  \setlength{\belowdisplayskip}{8pt plus 2pt minus 2pt}%
  \setlength{\abovedisplayshortskip}{3pt plus 2pt}%
  \setlength{\belowdisplayshortskip}{6pt plus 2pt minus 2pt}%
}
\setlist{topsep=4pt,itemsep=2pt,parsep=0pt}
\newtheorem{theorem}{Theorem}[section]
\newaliascnt{lemma}{theorem}
\newtheorem{lemma}[lemma]{Lemma}
\aliascntresetthe{lemma}
\newaliascnt{proposition}{theorem}
\newtheorem{proposition}[proposition]{Proposition}
\aliascntresetthe{proposition}
\newaliascnt{corollary}{theorem}
\newtheorem{corollary}[corollary]{Corollary}
\aliascntresetthe{corollary}
\theoremstyle{definition}
\newaliascnt{definition}{theorem}
\newtheorem{definition}[definition]{Definition}
\aliascntresetthe{definition}
\newaliascnt{example}{theorem}

\aliascntresetthe{example}
\theoremstyle{remark}
\newaliascnt{remark}{theorem}

\newtheorem*{remarkstar}{Remark}
\aliascntresetthe{remark}
\crefname{lemma}{Lemma}{Lemmas}
\crefname{proposition}{Proposition}{Propositions}
\crefname{corollary}{Corollary}{Corollaries}
\crefname{definition}{Definition}{Definitions}
\crefname{example}{Example}{Examples}
\crefname{remark}{Remark}{Remarks}
\crefname{appendix}{Appendix}{Appendices}
\Crefname{appendix}{Appendix}{Appendices}
\numberwithin{equation}{section}
\newcommand{\E}{\mathbb E}
\newcommand{\R}{\mathbb R}
\newcommand{\Prob}{\mathbb P}
\newcommand{\one}{\mathbf 1}
\newcommand{\OPT}{\operatorname{OPT}}
\newcommand{\Poi}{\operatorname{Poi}}

\newcommand{\gap}{\operatorname{gap}}

\newcommand{\clip}{\operatorname{clip}}

\newcommand{\calH}{\mathcal H}
\newcommand{\ip}[2]{\langle #1,#2\rangle}

\newcommand{\bconst}{\beta}
\newcommand{\aconst}{\alpha_0}

\hypersetup{pdftitle={A Sharp Barrier for Consistent Submodular Maximization: Any Improvement over 2-sqrt(2) Entails Exponential Queries or Linear Recourse},pdfauthor={Shi Fu; Qixin Zhang; Dacheng Tao},pdfsubject={Submodular maximization, oracle lower bounds, recourse, prices for unknown futures}}
\BeforeBeginEnvironment{theorem}{\Needspace{4\baselineskip}}
\BeforeBeginEnvironment{lemma}{\Needspace{3\baselineskip}}
\title{A Sharp Barrier of Consistent Submodular Maximization: Any Improvement over $2-\sqrt{2}$ Entails Exponential Queries or Linear Recourse}
\author{Shi Fu \qquad Qixin Zhang \qquad Dacheng Tao\\[4pt]
{\normalsize Nanyang Technological University, Singapore}}
\date{}
\begin{document}
\hypersetup{pageanchor=false}
\maketitle
\begin{abstract}
\normalsize
\noindent Consistent submodular maximization aims to maintain a high-utility 	subset of size at most $k$ as new elements arrive one by one, while changing only a limited number of elements per update.
Previously, \citep{DFL+25} established an information-theoretic bound of $2/3$ and gave the first polynomial-time randomized algorithm with a $0.51$-approximation guarantee.
However, a fundamental question remained open: 
\emph{Is it possible to design an efficient algorithm for consistent submodular maximization that matches the classical offline $(1-1/e)$-approximation?} In this paper, we answer this question in the negative by proving that the sharp approximation threshold achievable with polynomially many value queries and worst-case constant recourse is
\[
\bconst\triangleq2-\sqrt{2} \approx 0.5858 < 1-\frac{1}{e}.
\]\noindent Specifically, we develop an improved poly-time randomized algorithm, which can attain a $(\bconst-\varepsilon)$-approximation in expectation at each fixed time with $O(\varepsilon^{-2})$ worst-case recourse.
The cornerstone of this algorithm is a novel \emph{anchored sampling} scheme, which uses only value queries on subsets of the elements seen so far to construct a solution that remains competitive under unknown future arrivals.
Conversely, we establish a matching lower bound by showing that any fixed improvement over $\bconst$ requires either exponentially many value queries or $\Omega(k)$ changes per update. Furthermore, we  characterize the impact of curvature by establishing a sharp threshold of $1-(\sqrt{2}-1)\vartheta$ for consistent submodular maximization, where $\vartheta\in[0,1]$ denotes the total curvature. Surprisingly, for weighted coverage objectives, we devise a poly-time randomized algorithm that can match the optimal $(1-1/e-\varepsilon)$-approximation with $O(\varepsilon^{-1})$ worst-case recourse, thereby separating coverage maximization from general submodular maximization in the consistent settings.
\end{abstract}

\thispagestyle{empty}
\pagenumbering{gobble}
\clearpage
{%
  \hypersetup{linkcolor=black,linktoc=all}%
  \setcounter{tocdepth}{2}%
  \tableofcontents
}
\thispagestyle{empty}
\clearpage
\pagenumbering{arabic}
\hypersetup{pageanchor=true}
\section{Introduction}

\looseness=-1
%How much value must an algorithm lose when its solution must remain stable? This tradeoff has been studied in consistent clustering \citep{LV17,LHG+24} and dynamic submodular cover \citep{GL20}. In \emph{consistent submodular maximization}~\citep{DFL+24,DFL+25}, the task is to maintain a high-value selection as elements arrive. The objective is monotone and submodular, so adding elements never decreases value and marginal contributions diminish as the selected set grows. In weighted coverage, each element covers a collection of features, and a selection is valued by the total weight of the features it covers. A new element may overlap heavily with the current selection, making previously unselected elements better complements. A good solution after the arrival may therefore require replacing many old elements.

How much solution quality must be sacrificed when an algorithm's output is required to remain stable as the underlying problem evolves?
This tradeoff has been widely studied in various settings, including consistent clustering \citep{LV17,LHG+24} and dynamic submodular cover \citep{GL20}.
In \emph{consistent submodular maximization}~\citep{DFL+24,DFL+25}, the goal is to maintain a high-utility subset of size at most $k$ for a monotone submodular objective as new elements arrive one by one, while changing only a bounded number of elements per update.
The central challenge in this scenario is that a newly arriving element may overlap heavily with the current selection, potentially making several previously selected elements redundant. Consequently, obtaining a high-quality solution after its arrival may require replacing many previously selected elements, thereby violating the stability requirement. This tension motivates us to revisit consistent submodular maximization and investigate how the stability requirement affects the approximation achievable by computationally efficient algorithms.
%This raises a fundamental question: \emph{how can we maintain high-quality solutions while satisfying the stability constraint for submodular maximization?}
 %aiming to characterize the inherent tradeoff among approximation quality, computational efficiency and solution stability.

%An oblivious adversary fixes a finite ground set $V$, an arrival order, and a normalized monotone submodular objective $f:2^V\to\R_{\ge0}$, with $f(\varnothing)=0$, before the algorithm's random choices. Let $X_t$ be the first $t$ arrivals. For an integer capacity $k\ge1$, the algorithm maintains $S_t\subseteq X_t$ with $|S_t|\le k$, starting from $S_0=\varnothing$. It may store and revisit the full history. Its \emph{current value queries} evaluate $f(S)$ for any $S\subseteq X_t$, including sets larger than $k$.

%Approximation is measured in expectation at each fixed time, over the algorithm's randomness. The \emph{recourse} at time $t$ is the number of elements inserted into or removed from the maintained set, that is, $|S_t\triangle S_{t-1}|=|S_t\setminus S_{t-1}|+|S_{t-1}\setminus S_t|$, where $\triangle$ denotes symmetric difference. We require \emph{worst-case constant recourse}: the bound holds after every arrival and for every random outcome, and may depend on accuracy but not on $k$ or the stream length.

Formally, we consider a finite ground set $V$ whose elements arrive one by one, and a normalized monotone submodular objective $f:2^V\to\R_{\ge0}$.
Let $X_t$ denote the set of all elements that have arrived by time $t$. Given an integer capacity $k\ge1$, our goal is to design an algorithm that maintains a subset $S_t\subseteq X_t$ with $|S_t|\le k$ after each arrival, while satisfying two key requirements: approximation quality and solution stability.

First, we require the maintained solution to achieve high value. In particular, we say
that a randomized algorithm achieves an $\alpha$-approximation if $\mathbb{E}[f(S_t)]\ge\alpha\max_{S\subseteq X_t:\,|S|\le k}f(S)$ for all $t$, where the expectation is taken over the algorithm's randomness. Second, we require the maintained solution to evolve smoothly between consecutive updates.
Specifically, we define the \emph{recourse} at time $t$ as $|S_t\triangle S_{t-1}|=|S_t\setminus S_{t-1}|+|S_{t-1}\setminus S_t|$, where $\triangle$ denotes symmetric difference. We require $|S_t\triangle S_{t-1}|\le C$ for every time $t$, where $C$ is a prescribed constant. Throughout this paper, we suppose that the algorithm may store and revisit the value of $f$ on entire history of arrivals $X_t$.

%Without a recourse restriction, rerunning greedy after each arrival achieves the optimal offline $1-1/e$ coefficient with polynomially many current value queries \citep{NWF78,NW78}. Under constant recourse, the tight threshold with unrestricted queries and computation is $2/3$~\citep{DFL+25}, while the efficient guarantee approaches $0.51$~\citep[Theorem~6.8]{DFL+25}. The open question is whether efficient randomized algorithms can match $1-1/e$ with constant recourse~\citep[Section~1.1]{DFL+25}. We address the stronger oracle question, allowing unrestricted computation between polynomially many queries. A query lower bound then rules out polynomial-time value-oracle algorithms without a computational hardness assumption.

%We prove that the supremum of the approximation ratios achievable with polynomially many current value queries and worst-case constant recourse is
%\[
%\bconst:=2-\sqrt{2}\approx0.5858<1-\frac{1}{e}.
%\]
%For every $\varepsilon>0$, our randomized algorithm achieves a $\bconst-\varepsilon$ approximation in expectation at each fixed time, with $O(\varepsilon^{-2})$ worst-case recourse. Conversely, any fixed improvement over $\bconst$ requires $\exp(\Omega(k))$ value queries before a critical arrival or $\Omega(k)$ changes at that arrival, even if subsequent queries and computation are unrestricted. The gap below $1-1/e$ therefore comes from limiting queries and recourse together.

Broadly speaking, rerunning the greedy algorithm after each arrival achieves the optimal offline $(1-1/e)$-approximation guarantee using only polynomially many current value queries \citep{NWF78,NW78}.
However, the greedy solution may change drastically upon the arrival of a single element, potentially requiring $\Omega(k)$ changes and thus violating the constant-recourse requirement~\citep{DFL+24}. To cross this hurdle, \citep{DFL+25} established an information-theoretic bound of $2/3$ for consistent submodular maximization and presented the first poly-time randomized algorithm achieving a $(0.51-O(\varepsilon))$-approximation with $O(\varepsilon^{-2})$ worst-case recourse \citep[Theorem~6.8]{DFL+25}. Recall that the classical cardinality-constrained monotone submodular maximization admits an optimal $(1-1/e)$-approximation~\citep{Fei98}.
In light of this benchmark, a fundamental question naturally arises: \emph{can we design an efficient algorithm for consistent submodular maximization while matching the classical offline $(1-1/e)$-approximation ratio?}~\citep[Section~1.1]{DFL+25}

The core of this paper is to answer this question in the negative by establishing that the sharp approximation threshold achievable with polynomially many current value queries and worst-case constant recourse is\[
\bconst\triangleq2-\sqrt{2}\approx0.5858<1-\frac{1}{e}.
\]

%\looseness=-1
%Both bounds concern what current queries reveal about which old elements a future arrival will make useful. Our algorithm uses greedy values and residual bounds to mix anchored samples from current greedy prefixes. After adjoining any fixed future set, the sample's expected value competes with the best such union under the same limit on current elements. A supporting-line argument yields $\bconst$ (\cref{sec:ideas-threshold}). The checkpoint framework~\citep[Theorem~3.2 and Section~4]{DFL+25} converts this guarantee into online solutions with constant recourse. The lower-bound construction enforces exact hiding and submodularity. Before the critical arrival, exponentially many adaptive exact queries cannot reliably identify enough useful old elements. Afterward, each inserted old element adds only an $O(1/k)$ fraction of optimum, forcing $\Omega(k)$ insertions for a fixed improvement even if the useful elements are revealed.

More specifically, our main contributions are two-fold. On the one hand, we develop a polynomial-time randomized algorithm that achieves a $(\bconst-\varepsilon)$-approximation in expectation at each fixed time with $O(\varepsilon^{-2})$ worst-case recourse.
The key ingredient of our algorithm is an \emph{anchored sampling} scheme built upon the standard greedy algorithm, which can construct a good subset whose expected value remains robust and competitive for unknown future arrivals. Subsequently, the checkpoint framework~\citep[Theorem~3.2 and Section~4]{DFL+25}  converts this robust sampling guarantee into a dynamic $(\bconst-\varepsilon)$-approximation algorithm with $O(\varepsilon^{-2})$ worst-case constant recourse.

\looseness=-1

On the other hand, we construct a carefully designed monotone submodular objective with a critical arriving element. When this element arrives, our construction ensures that, with high probability, any algorithm making only polynomially many value queries will be effectively trapped near a $\bconst$-approximation subset. Furthermore, we show that any fixed improvement beyond this aforementioned $\bconst$-approximation solution will require $\Omega(k)$ updates. Taken together, we naturally obtain a matching information-theoretic lower bound.

\subsection{Our Results}\label{sec:results}

Let $\OPT_k(X_t)=\max\{f(O):O\subseteq X_t,\ |O|\le k\}$ denote the available optimum. \Cref{tab:main-results} compares the guarantees. Computational conventions are specified in \cref{sec:model-cores}.

\begin{table}[!t]
\centering
\caption{Efficient approximation and hardness bounds for consistent submodular maximization.}
\label{tab:main-results}
\small
\renewcommand{\arraystretch}{1.3}
\setlength{\tabcolsep}{3pt}
\begin{tabular*}{\linewidth}{@{\extracolsep{\fill}}lccccc}
\toprule
\shortstack[l]{Function\\class}
& \shortstack{Prior\\approximation}
& \shortstack{Offline\\hardness}
& \shortstack{Our\\approximation}
& \shortstack{Our\\recourse}
& Inapproximability \\
\midrule
General
& $0.51-\varepsilon$~\citep{DFL+25}
& $1-1/e$~\citep{Von13}
& \cellcolor{accentlight}$\boldsymbol{2-\sqrt2-\varepsilon}$
& \cellcolor{accentlight}$O(\varepsilon^{-2})$
& \cellcolor{accentlight}$\boldsymbol{2-\sqrt2}^{\,\dagger}$ \\
Coverage
& $0.51-\varepsilon$~\citep{DFL+25}
& $1-1/e$~\citep{Fei98}$^{\ddagger}$
& \cellcolor{accentlight}$\boldsymbol{1-1/e-\varepsilon}$
& \cellcolor{accentlight}$O(\varepsilon^{-1})$
& $1-1/e$~\citep{Fei98}$^{\ddagger}$ \\
\bottomrule
\end{tabular*}
\par\smallskip
\begin{minipage}{\linewidth}
\small
Our algorithms use polynomially many current value queries per arrival.
Our approximation and recourse guarantees are given in
Theorem~\ref{thm:main}(i) for general functions and
Theorem~\ref{thm:coverage} for coverage.
$^\dagger$\,Theorem~\ref{thm:main}(ii) gives a matching
query--recourse lower bound, even with unrestricted computation.
$^\ddagger$\,Coverage hardness assumes
$\mathsf{NP}\not\subseteq\mathsf{BPP}$
and holds even without recourse restrictions.
\end{minipage}
\end{table}

\begin{theorem}[Sharp query--recourse threshold]\label{thm:main}
Let $\bconst=2-\sqrt2$.

\smallskip
\noindent\textup{\textbf{(i) Algorithmic guarantee.}} For every $\varepsilon\in(0,\bconst)$, there is a randomized value-oracle algorithm that makes $O(|X_t|k)$ current value queries per arrival, performs polynomially many additional arithmetic operations, and satisfies, for every $t$,
\begin{equation}\label{eq:model}
\E f(S_t)\ge(\bconst-\varepsilon)\OPT_k(X_t),
\qquad |S_t\triangle S_{t-1}|=O(\varepsilon^{-2}).
\end{equation}
The recourse bound holds for every outcome of the algorithm's randomness.

\smallskip
\noindent\textup{\textbf{(ii) Matching lower bound.}} For every $\zeta>0$ there are constants $c_\zeta,d_\zeta>0$ such that the following holds for all sufficiently large integers $k$. Consider any randomized value-oracle algorithm that, on every instance of this model with at most $d_\zeta k$ elements and every outcome of its randomness, makes at most $e^{c_\zeta k}$ queries before the last arrival $t$ and satisfies $|S_t\triangle S_{t-1}|\le c_\zeta k$. Then there exists an instance with a stream of at most $d_\zeta k$ elements and a normalized monotone submodular objective $f$, both fixed independently of the algorithm's random choices, such that at the final arrival~$t$,
\[
\E f(S_t)<(\bconst+\zeta)\OPT_k(X_t).
\]
Queries and computation after that arrival may be unrestricted.
\end{theorem}

Part~(ii) gives an unconditional separation between polynomially many queries and unrestricted querying under constant recourse. Their thresholds are $\bconst$ and $2/3$, respectively. The hard objective is explicit given its hidden partition, and the barrier concerns what can be learned before the last arrival and changed afterward. \Cref{thm:finite-lower} gives the explicit tradeoff in $k$.

\paragraph{One geometry for both bounds.}\label{par:one-geometry}
The normalized continuous initial segment of the hidden-color construction in \citep[Lemma~2.1 and Definition~2.3]{HKMY22} meets our supporting-line equality conditions at the contact parameter $T=\sqrt2$ (\cref{sec:ideas-threshold,fig:profiles}). Our finite construction adds a smooth continuation and explicit corrections to ensure exact hiding under full-history queries, final-arrival compatibility, and small marginals against every background set. The limiting ratio is $\bconst=2-\sqrt2$.

\paragraph{Weighted coverage.}
For weighted coverage, we improve both the approximation guarantee and the recourse bound of \cref{thm:main}.

\begin{theorem}[Informal version of \cref{thm:coverage}]\label{thm:intro-coverage}
Suppose that the objective on the entire stream is a nonnegative weighted coverage function. For every $\varepsilon>0$, there is a randomized algorithm that achieves a $1-1/e-\varepsilon$ approximation in expectation at each fixed time with $O(\varepsilon^{-1})$ worst-case symmetric recourse. It uses polynomially many current value queries per arrival without access to the underlying features or weights.
\end{theorem}

Consistency distinguishes two classes that share the offline $1-1/e$ guarantee. Coverage preserves this guarantee up to $\varepsilon$, whereas general objectives fall to $\bconst$ when both queries and recourse are limited. This separation is unconditional. Coverage's computational optimality is a separate statement: by Feige's maximum-coverage hardness \citep{Fei98}, any fixed improvement over $1-1/e$ in randomized polynomial time on explicit instances would imply $\mathsf{NP}\subseteq\mathsf{BPP}$, even without recourse limits. With unrestricted queries and computation, the coverage threshold under constant recourse is $3/4$~\citep[Theorems~5.5 and~5.7]{DFL+25}.

\paragraph{Bounded curvature.}
Following \citep{CC84}, $f$ has total curvature at most $\vartheta\in[0,1]$ if every marginal satisfies $f(S\cup\{i\})-f(S)\ge(1-\vartheta)f(\{i\})$ for all $S$ and $i\notin S$. The bound $\vartheta=1$ imposes no restriction, while $\vartheta=0$ means that $f$ is modular.

\begin{theorem}[Informal version of \cref{thm:curvature-main}]\label{thm:intro-curvature}
Suppose that the objective on the entire stream has total curvature at most a known $\vartheta\in[0,1]$, and let $\rho_\vartheta=1-(\sqrt2-1)\vartheta$. For every $\varepsilon>0$, there is a randomized algorithm that achieves a $\rho_\vartheta-\varepsilon$ approximation in expectation at each fixed time with $O(\varepsilon^{-2})$ worst-case symmetric recourse, using polynomially many current value queries per arrival. The threshold is sharp. For every fixed $\vartheta>0$, any fixed improvement over $\rho_\vartheta$ requires $e^{\Omega(k)}$ value queries before a critical arrival or $\Omega(k)$ changes at that arrival.
\end{theorem}

The threshold is $\bconst$ at $\vartheta=1$, while at $\vartheta=0$ symmetric recourse two maintains an exact optimum. Relative to the optimal offline threshold $1-\vartheta/e$ with polynomially many value queries~\citep{SVW17}, the cost of consistency is exactly $(\sqrt2-1-1/e)\vartheta\approx0.046\vartheta$, vanishing linearly as $\vartheta\to0$. Coverage can have curvature one, where $\rho_1=\bconst<1-1/e$, so \cref{thm:intro-coverage} does not follow from \cref{thm:intro-curvature}.

Both structural results use \emph{prices}, weights on current elements computed from current values alone. Their linear inequalities compare fractional solutions with feasible current sets after adjoining any fixed future set permitted by the structural assumption (\cref{sec:prices}).

\subsection{Related Work}\label{sec:related-work}

\noindent\textbf{Relation to prior work on consistency.}
The model originates in \citep{DFL+24}. We reuse the checkpoint reduction and future-comparison benchmark of \citep[Theorem~3.2 and Section~5.1, Equation~(8)]{DFL+25}. Our anchored sampler and current-only linear program raise their efficient addition-robust guarantee from $0.51$ to $\bconst$. \Cref{sec:ideas-sampling} explains the algorithmic difference, and \cref{sec:ideas-threshold} proves the sharper coefficient.

Their $2/3$ lower bound uses objectives that agree on every subset of current elements and differ only when the final element is present \citep[Section~5.2]{DFL+25}. Thus current queries cannot distinguish the future scenario, regardless of their number. We reuse the final-arrival template, but our current function already encodes a useful $k$-set. We must hide it from exponentially many adaptive exact queries, even on sets larger than $k$, while preserving submodularity and limiting repair after the arrival (\cref{sec:ideas-lower}). The fully dynamic framework of \citep{DFL+26} also handles deletions.

\noindent\textbf{Dynamic query complexity.}
Chen and Peng~\citep[Theorem~1.2]{CP22} prove a query lower bound for fully dynamic maximization, using deletions and allowing unrestricted recourse. Our bound uses insertion-only streams and limits changes at one arrival, even with unlimited queries afterward. Without recourse restrictions, their insertion-only algorithm achieves $1-1/e-\varepsilon$ with polynomially many queries~\citep[Theorem~1.3]{CP22}. Together with the unrestricted $2/3$ threshold above, this shows that neither restriction alone forces an approximation below $1-1/e$.

\noindent\textbf{Streaming and coresets.}
The constant $2-\sqrt2$ occurs in an element-store streaming space lower bound~\citep[Theorem~1.2]{HKMY22} and randomized composable coreset bounds~\citep[Theorems~4.1 and~4.8]{MZ15}. These concern storage or the quality of summary unions. Our model permits full-history access and restricts queries and recourse.

\noindent\textbf{Other models.}
Related models include \textsc{Swapping} in streaming~\citep{CK15}, preemption and free disposal~\citep{BFS19,CHJ+18}, consistent online optimization~\citep{KJKLV19}, and adversarial-injection streaming~\citep{GKRS20,WCJ+26}. Detailed comparisons, including the parameters of the dynamic query bounds, appear in \cref{app:further-related}.

\paragraph{Organization.}
\Cref{sec:preliminaries} defines robust cores and states the checkpoint conversion from \citep{DFL+25}. \Cref{sec:ideas} explains the main ideas, followed by the general upper and lower bounds in \cref{sec:cores,sec:lower}. Prices in \cref{sec:prices} support coverage and curvature in \cref{sec:structure,sec:curvature}. \Cref{sec:conclusion} discusses open problems.

\section{Preliminaries}\label{sec:preliminaries}

\subsection{Model and Robust Cores}\label{sec:model-cores}

We use the model and notation from the introduction, without assuming advance knowledge of the stream length. Each value query has unit cost, as in \citep[Section~2]{DFL+25}. Our query--recourse lower bound is information-theoretic and imposes no restriction on computation time. Our upper bounds use exact arithmetic and exact sampling from the specified distributions. Algorithms must terminate almost surely after every arrival on every valid instance.

For $A,B\subseteq V$, let $f(A\mid B)$ denote the marginal gain $f(A\cup B)-f(B)$, with $f(i\mid B)=f(\{i\}\mid B)$. A function $f$ is \emph{monotone} if $f(i\mid S)\ge0$ for all $S\subseteq V$ and $i\in V\setminus S$, and \emph{submodular} if $f(i\mid S)\ge f(i\mid T)$ whenever $S\subseteq T\subseteq V$ and $i\in V\setminus T$. A \emph{weighted coverage function} has the form $f(S)=\sum_{a\in\mathcal U}w_a\one\{S\cap C_a\ne\varnothing\}$ for a finite feature set $\mathcal U$, nonnegative weights $w_a$, and covering sets $C_a\subseteq V$. The algorithm accesses only function values. Structural promises hold on all of $V$, and any curvature bound is known to the algorithm.

For the static subroutine, fix a current set $X$ and integer capacity $\kappa\ge1$, reserving part of the online capacity $k$ for later arrivals. Each fixed future set $R\subseteq V\setminus X$ induces $h_R(S)=f(S\cup R)$ on $2^X$, with $h_R(\varnothing)=f(R)$ possibly positive. Monotonicity gives $h_R\ge f$, and submodularity bounds its marginals by the corresponding current marginals of $f$. We restate the framework of \citep[Section~5.1, Equation~(8)]{DFL+25} using the names \emph{legal future} and \emph{robust core} for convenience.

\begin{definition}[Legal future]\label{def:legal-future}
A monotone submodular function $h:2^X\to\R_{\ge0}$ is a \emph{legal future} if $h(S)\ge f(S)$ for all $S\subseteq X$ and
\begin{equation}\label{eq:legal}
0\le h(S\cup\{i\})-h(S)\le f(i\mid S)
\qquad(S\subseteq X,\ i\in X\setminus S).
\end{equation}
\end{definition}
\Needspace{5\baselineskip}
\begin{definition}[Robust core]\label{def:robust-core}
For a legal future $h$, let $P_h=\max\{h(O): O\subseteq X,\ |O|\le\kappa\}$. For $0<\alpha\le1$, a procedure using only current value queries is an \emph{$\alpha$-robust core} if it outputs a random set $A\subseteq X$, with $|A|\le\kappa$ on every outcome, satisfying
\begin{equation}\label{eq:robust-core-interface}
\E h(A)\ge\alpha P_h
\qquad\text{for every legal future }h.
\end{equation}
\end{definition}
The law of $A$ is chosen before $h$, with a separate expected guarantee for each fixed $h$. For $h=h_R$, \cref{eq:robust-core-interface} is the guarantee of \citep[Equation~(8)]{DFL+25}: both sample and comparison set adjoin $R$ without charging it to $\kappa$. This implies their addition robustness \citep[Definition~3.1]{DFL+25}: for every fixed $R\subseteq V\setminus X$, $\E f(A\cup R)\ge\alpha\max_{U\subseteq X\cup R,\,|U|\le\kappa}f(U)$, since every such $U$ satisfies $f(U)\le h_R(U\cap X)\le P_h$. Recent arrivals occupy reserved positions within $k$ in \cref{sec:checkpoint}.

\subsection{The Checkpoint Conversion}\label{sec:checkpoint}

The Check-Point reduction of \citep[Section~4]{DFL+25} turns a robust core into an online algorithm with small recourse. At regular checkpoints, it samples a new set from the current prefix and gradually replaces the previous sample during a short random window. Outside this window, the maintained set contains a complete sample together with all arrivals $R$ since its checkpoint, so \cref{eq:robust-core-interface} applies to the legal future $h_R(S)=f(S\cup R)$.

\begin{lemma}[Checkpoint conversion, adapted from {\citep[Theorem~3.2]{DFL+25}}]\label{lem:checkpoint}
Let $0<\alpha\le1-\nicefrac{1}{e}$. Suppose an $\alpha$-robust core uses polynomially many current value queries for every current set $X$ and capacity $\kappa\ge1$. For every integer $B\ge4$, there is an online algorithm achieving approximation $\alpha(1-2/B)(1-1/B)$ with worst-case recourse at most $8B^2+2$. It calls the core or greedy at most once per arrival.
\end{lemma}

The factor $1-2/B$ comes from reserving capacity for recent arrivals, with submodularity controlling the loss in optimum. Intermediate mixtures need not satisfy the core guarantee, so the replacement window is chosen independently of the samples, with any fixed time lying inside it with probability at most $1/B$, giving the factor $1-1/B$. Appendix~\ref{app:checkpoint-details} gives the full schedule and proof.

\section{Techniques and Ideas}\label{sec:ideas}

\looseness=-1
Which old elements should a stable solution keep before their future usefulness is known? By \cref{lem:checkpoint}, it suffices to construct a robust core using only current value queries. We first develop greedy certificates and anchored sampling in \cref{sec:ideas-sampling}, then use supporting-line geometry in \cref{sec:ideas-threshold} to show that the current-only LP has value at least $\bconst=2-\sqrt2$. We next explain in \cref{sec:ideas-lower} why this threshold is sharp under polynomially many current queries and constant recourse, before showing in \cref{sec:ideas-structure} how coverage and bounded curvature permit stronger guarantees.

\subsection{Greedy Certificates and Anchored Sampling}\label{sec:ideas-sampling}

Our goal is to improve the $0.51$ guarantee of \citep[Theorem~6.2]{DFL+25} using certificates along a full greedy chain. In their perfect-alignment example \citep[Example~3.3]{DFL+25}, adapted to our benchmark, $2\kappa$ current elements cover distinct unit-weight features. A future covering a deterministic choice's $\kappa$ features gives that choice value $\kappa$ and its complement value $2\kappa$, with the future adjoined without using capacity. Thus beating $1/2$ requires randomization.

\noindent\textbf{Two greedy certificates.}
Let $G_j$ contain the real elements in the first $j$ greedy positions, with null padding as in \cref{sec:core-construction}, and let $v_j=f(G_j)$ and $D_j=\kappa(v_{j+1}-v_j)$. Every fixed legal future $h$ satisfies
\[
h(G_j)\ge v_j,\qquad h(G_j)\ge P_h-D_j.
\]
The first bound uses $h\ge f$. For the second, legality and submodularity bound the remaining contribution of a feasible comparison set by its current marginals, each at most $D_j/\kappa$ by the greedy choice. This is the classical greedy value--residual argument \citep{NWF78}. Since $v_j$ is nondecreasing and $D_j$ is nonincreasing, both certificates favor longer prefixes, which may exceed capacity.

\noindent\textbf{Anchored sampling.}
For $0\le a<\kappa<b\le2\kappa$, the sampler $A_{a,b}$ keeps $G_a$ and uniformly samples $\kappa-a$ positions from $a+1,\ldots,b$, discarding nulls. By \cref{lem:anchor},
\begin{equation}\label{eq:ideas-anchor}
\E h(A_{a,b})\ge(1-\lambda)h(G_a)+\lambda h(G_b),
\qquad \lambda=\frac{\kappa-a}{b-a}.
\end{equation}
This is the sampling principle used in \citep[Lemma~6.3(iii)]{DFL+25}. Varying the anchor makes every endpoint chord strictly crossing $\kappa$ realizable by feasible samples (\cref{fig:anchor}).

\begin{figure}[!t]
\centering
\begin{tikzpicture}[
  font=\small, y=0.93cm,
  axis/.style={-{Latex[length=1.8mm]}, line width=0.5pt, draw=slate},
  node base/.style={circle, minimum size=7.4mm, inner sep=0pt, line width=0.7pt},
  kept/.style={node base, fill=accentblue, draw=accentblue, text=white},
  drawn/.style={node base, fill=accentblue!14, draw=accentblue},
  skipped/.style={node base, fill=white, draw=slate!60, text=slate},
  brace/.style={decorate, decoration={brace, mirror, amplitude=4pt}, draw=slate!80, line width=0.5pt},
  lab/.style={font=\small, text=slate},
  result/.style={font=\small, text=black},
]
% ---------------- (a) Anchored sampling
\node[anchor=base west, font=\bfseries\small] at (0,2.55) {(a) Anchored sampling};
\node[kept]    (g1) at (0.45,1.45) {$g_{1}$};
\node[kept]    (g2) at (1.35,1.45) {$g_{2}$};
\node[drawn]   (g3) at (2.75,1.45) {$g_{3}$};
\node[skipped] (g4) at (3.65,1.45) {$g_{4}$};
\node[drawn]   (g5) at (4.55,1.45) {$g_{5}$};
\node[skipped] (g6) at (5.45,1.45) {$g_{6}$};
\draw[brace] (0.08,0.93) -- (1.72,0.93) node[midway, below=5pt, lab] {anchor $G_2$};
\draw[brace] (2.38,0.93) -- (5.82,0.93) node[midway, below=5pt, lab] {sample 2 of 4};
\node[anchor=base, result] at (2.95,-0.40) {$|A_{2,6}|=\kappa=4$};

% ---------------- (b) Chord bound
\begin{scope}[shift={(7.35,0)}]
\node[anchor=base west, font=\bfseries\small] at (-0.15,2.55) {(b) Chord bound};
\draw[axis] (0,0.45) -- (6.05,0.45) node[right, lab, font=\footnotesize] {$j$};
\draw[axis] (0,0.45) -- (0,2.35);
\coordinate (A) at (1.60,1.00);
\coordinate (K) at (3.10,1.525);
\coordinate (B) at (4.60,2.05);
\foreach \p in {A,K,B} \draw[densely dashed, slate!55, line width=0.45pt] (\p) -- (\p |- 0,0.45);
\draw[accentblue, line width=1.1pt] (A) -- (B);
\fill[accentblue] (A) circle (2.1pt) (B) circle (2.1pt);
\filldraw[accentblue, fill=white, line width=0.9pt] (K) circle (2.4pt);
\node[above=3pt] at (A) {$h(G_2)$};
\node[above=3pt] at (B) {$h(G_6)$};
\foreach \x/\t in {1.60/2, 3.10/\kappa=4, 4.60/6}
  \node[anchor=base, font=\small] at (\x,0.10) {$\t$};
\node[anchor=base, result] at (3.0,-0.40) {$\E\,h(A_{2,6})\ge\tfrac12 h(G_2)+\tfrac12 h(G_6)$};
\end{scope}
\end{tikzpicture}
\caption{Anchored sampling with the first six greedy elements and $\kappa=4$. (a)~Keep $G_2$ (dark) and sample two of the next four elements uniformly (light shows one outcome). (b)~The open marker gives the chord lower bound on $\E h(A_{2,6})$ for every fixed legal future $h$ (\cref{lem:anchor}).}
\label{fig:anchor}
\end{figure}

In \cref{fig:anchor}, suppose the six elements cover distinct unit-weight features. A future element covering the last four features gives expected values $6$ for $A_{2,6}$ and $16/3$ for a uniform four-subset of $G_6$, whereas one covering the first four gives $5$ and $16/3$. The current values are identical, yet the ranking reverses: anchoring provides different protection across futures. The $0.51$ algorithm samples uniformly from a threshold-augmented greedy set: it first constructs a greedy $\kappa$-set, then adds elements above a fixed marginal threshold \citep[Section~6]{DFL+25}. Our chain continues maximum-marginal greedy, and we vary both endpoints to mix their certificates.

\noindent\textbf{Current-only mixing.}
Prior work already uses LPs and duality for distributions committed before the future is known \citep[Section~5.1]{DFL+25}. Here we use a polynomial-size family of current greedy-chain certificates, without enumerating future scenarios. The endpoint bounds are affine in $P_h$, with coefficient pairs $(0,v_j)$ and $(1,-D_j)$ for value and residual. A \emph{labeled action} $\pi$ specifies an anchored sampler, or $G_\kappa$, and a certificate at each endpoint. Interpolation gives $\E h(A_\pi)\ge\mu_\pi P_h+\nu_\pi$, so a mixture $p$ satisfies
\begin{equation}\label{eq:ideas-mixture}
\E h(A)\ge
\underbrace{\Bigl(\sum_\pi p_\pi\mu_\pi\Bigr)}_{\mathclap{\text{average slope}}}P_h
+\underbrace{\sum_\pi p_\pi\nu_\pi}_{\mathclap{\text{average intercept}}}.
\end{equation}
The LP maximizes the average slope subject to a nonnegative average intercept: value certificates offset the negative intercepts of residual certificates. The same current-only mixture guarantees this coefficient in expectation for every fixed legal future. An optimum uses at most two labeled actions (\cref{sec:cores}). We next prove that its coefficient is at least $\bconst$.

\subsection{Supporting-Line Geometry and the Sharp Threshold}\label{sec:ideas-threshold}

We prove that the current-only LP in \cref{sec:ideas-sampling} has value at least $\bconst=2-\sqrt2$. If the residual bound $P_h-D_j$ remains weak, then $D_j$ must be large, while $v_{j+1}-v_j=D_j/\kappa$ forces the greedy value to accumulate along the chain. The supporting-line argument makes this tradeoff precise.

First, normalize prefix size as $q_j=j/\kappa$. The LP constraint requiring a nonnegative average intercept introduces a dual multiplier $\omega\ge0$, giving endpoint heights $u_j=\max\{\omega v_j,1-\omega D_j\}$. Anchored sampling, together with the action $G_\kappa$, realizes the upper concave envelope of the points $(q_j,u_j)$ at $t=1$. LP duality identifies the program value with $\inf_{\omega\ge0}\gamma(\omega)$, where $\gamma(\omega)$ is this envelope height. The envelope admits a supporting line $c+dt$ with $c,d\ge0$ and $c+d=\gamma(\omega)$.

Next, fix $\omega\ge0$ and suppose for a contradiction that $\gamma=\gamma(\omega)<\bconst$. The linear interpolation $s(q_j)=\omega v_j$ lies below $c+dt$, while the residual bound gives $s'(t)\ge1-c-dt$ on each grid cell. Integrating to $2$ yields $d\ge2-3\gamma>0$, which places $t_*=(1-\gamma)/d$ in $(0,2)$. Integrating to $t_*$ and comparing with the line gives
\begin{equation}\label{eq:ideas-beta}
(1-\gamma)^2\le2cd\le\frac{\gamma^2}{2},
\end{equation}
where the second inequality uses $c+d=\gamma$. This contradicts $\gamma<2-\sqrt2$, so every dual value, and hence the LP optimum, is at least $\bconst$. \Cref{thm:anchored} handles the grid exactly. The resulting $\bconst$-robust core and \cref{lem:checkpoint} give $\bconst-\varepsilon$ online approximation with $O(\varepsilon^{-2})$ recourse. This scalar tradeoff also appears in \citep[Lemma~4.6, Eq.~(4)]{MZ15}, and \cref{sec:core-threshold} gives the correspondence.

\Needspace{9\baselineskip}
At $T=\sqrt2$, the normalized continuous initial segment of \citep[Lemma~2.1 and Definition~2.3]{HKMY22} satisfies the equality conditions in \eqref{eq:core-quadratic}: $c=d=\bconst/2$ and $v(t)=(1-c)t-dt^2/2$ on $[0,T]$. With $w(t)=c+dt=1-v'(t)$ there, the profiles meet tangentially at $T$ (\cref{fig:profiles}), giving balanced value $w(1)=\bconst$ and hidden-complement value $c+v'(0)=1$. Tangency at a general $T$ gives ratio $\mathcal R(T)=(T^2+2)/(T^2+2T+2)$, minimized at $\sqrt2$. \Cref{sec:lower} independently constructs the continuation, exact hiding, and final-arrival compatibility, obtaining finite hard instances that approach this ratio.

\begin{figure}[!t]
\centering
\includegraphics[width=.78\textwidth]{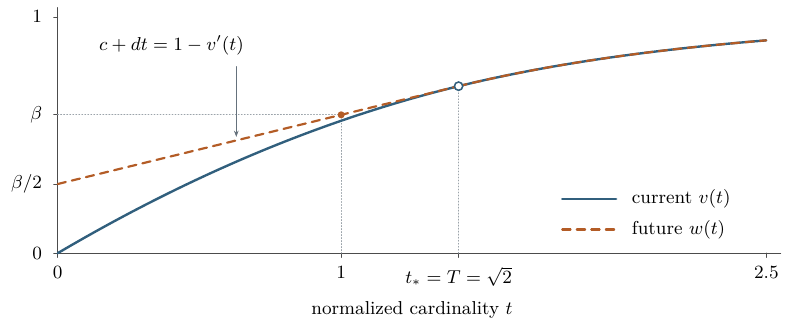}
\caption{One limiting geometry for both bounds, with $c=d=\bconst/2$ and $T=\sqrt2$. The continuous curve $v$ is a limiting equality profile for the finite-grid interpolation $s$, with supporting line $c+dt$. For the lower bound, the profiles $v,w$ give balanced value $w(1)=\bconst$ and hidden-complement value $1$.}
\label{fig:profiles}
\end{figure}

\subsection{Matching Lower Bound: Exact Hiding and Costly Repair}\label{sec:ideas-lower}

We now show why no fixed improvement over $\bconst=2-\sqrt2$ is possible with polynomially many current value queries and constant recourse. Exact hiding prevents the algorithm from reliably selecting enough useful old elements before a critical arrival, while small marginals limit the gain from later insertions even if the hidden set is revealed. Both obstacles must arise from the same monotone submodular objective.

\looseness=-1
We adapt the inherited continuous initial segment described in \cref{sec:ideas-threshold}, partitioning the current set into a hidden $k$-set $A$ and $mk$ other elements. Following the final-arrival template of \citep[Section~5.2]{DFL+25}, we introduce an element \(r\) that makes \(A\) the useful complement. The objective $f_A$ depends only on group counts and the presence of $r$. \emph{Legality} requires monotonicity and submodularity in both branches and a nonnegative diminishing marginal for $r$, even on sets larger than $k$. \emph{Exact hiding} requires current answers to depend only on query size in a band around balanced counts, despite access to the full arrived history. Approximate agreement would not suffice against an exact oracle.

The exact-flattening and concave-regularization principles in \citep[proof of Lemma~3.2]{Von13} provide the starting point. Here the two branches must jointly define one submodular objective, so their gradients must remain ordered. Before flattening, a hidden element has the same marginal with and without $r$. Flattening the current branch can therefore reverse this ordering, making an old element more valuable after $r$ arrives. We correct only the future branch to restore the ordering, relying on derivative margins established by a common exponential regularizer. 

\noindent\textbf{Before the arrival: hiding extends to the output.}
For an algorithm making at most $Q$ current queries, we adapt the reference-transcript argument from the proof of \citep[Lemma~3.3]{Von13}. Run it against a valid oracle whose answers depend only on query size. For almost every fixed random tape, its queries and pre-arrival output are fixed independently of $A$. Concentration and a union bound over at most $Q+1$ sets, including the possibly unqueried output, show that the true execution agrees with the reference one and its output lies in the balanced band, except with probability $O((Q+1)e^{-\Omega(k/m^2)})$.

\noindent\textbf{After the arrival: information cannot replace changes.}
Let $S^-,S^+$ be the outputs before and after $r$, and let $C$ bound symmetric recourse at that arrival. Each old element has marginal at most $5/k$ against every background set. Even if $A$ is revealed and all subsequent queries and computation are unrestricted, granting $r$ for free gives, on the hiding event,
\begin{equation}\label{eq:ideas-repair}
f_A(S^+)\le
\underbrace{f_A(S^-\cup\{r\})}_{\le\,\mathcal R(T)+O(1/m)}
+\underbrace{5C/k}_{\text{repair}}.
\end{equation}
By contrast, the feasible set $\{r\}\cup(A\setminus\{a_0\})$ has value at least $1-5/k$ for any $a_0\in A$. Taking $T=\sqrt2$ gives $\mathcal R(T)=\bconst$ exactly. For any fixed target improvement, choose $m$ large enough to control the $O(1/m)$ error, then take $k$ sufficiently large. \Cref{thm:finite-lower} gives the resulting tradeoff: the improvement requires $e^{\Omega(k)}$ queries before the arrival or $\Omega(k)$ changes at that arrival.

\subsection{Beyond the General Barrier: Prices and Structure}\label{sec:ideas-structure}

The general threshold $\bconst$ need not govern more structured objectives. We next explain how current-only prices exploit additional structure to improve the approximation or recourse guarantee.

\noindent\textbf{Prices computed from current values.}
For $\boldsymbol{x}\in[0,1]^X$ with $\sum_i x_i\le\kappa$, let $Z_{\boldsymbol{x}}$ include each $i$ independently with probability $1-e^{-x_i}$ and define the \emph{response} $H_h(\boldsymbol{x})=\E h(Z_{\boldsymbol{x}})$. Although $Z_{\boldsymbol{x}}$ may exceed capacity, rounding gives a feasible sample with at least this expected value.

\Needspace{5\baselineskip}
A price vector $\boldsymbol{p}=\boldsymbol{p}(\boldsymbol{x})$, computed from current values, gives the linear comparison bound
\[
H_h(\boldsymbol{x})-\alpha h(O)\ge\langle\boldsymbol{p},\boldsymbol{x}-\one_O\rangle
\]
for every feasible comparison set $O$ and legal future $h$ allowed by the structural promise. If $\boldsymbol{x}$ nearly maximizes the linear objective defined by its own prices, the right side is nearly nonnegative for every $O$, certifying a uniform approximation guarantee over the allowed futures. \Cref{sec:prices} constructs prices as gradients of a potential using only current values. Integrating gradients across sampling scales yields coefficient $\bconst$ for every legal future, extending the approaches in \citep{FW14} and \citep{ZDC+22}.

\noindent\textbf{Coverage protects intermediate solutions.}
For coverage, the gradient of the current response gives coefficient $1-1/e$ using only value queries, without features or weights. The future response $H_R(\boldsymbol{x})=\E f(Z_{\boldsymbol{x}}\cup R)$ is concave in $\boldsymbol{x}$ for each fixed $R$. Sampling through $\kappa$ independent categorical slots, each choosing an element or nothing, ensures feasibility on every outcome. Concavity then makes gradual replacement preserve the interpolation of the old and new responses in expectation. Thus every intermediate state is protected, and the random window of \cref{lem:checkpoint} is unnecessary. Replacement throughout the block reduces per-arrival changes from $O(B^2)$ to $O(B)$, giving a $1-1/e-\varepsilon$ approximation with $O(\varepsilon^{-1})$ worst-case recourse (\cref{lem:slot-conversion}).

\Needspace{7\baselineskip}
\noindent\textbf{Curvature preserves modular value.}
Under a known curvature bound $\vartheta$, write $f=\ell+g$ with $\ell_i=(1-\vartheta)f(\{i\})$ and monotone submodular $g$. Following \citep{SVW17}, we keep the modular term linear and apply prices only to $g$, comparing against $\ell(O)+\bconst g(O)$ for the same feasible set $O$. Thus the modular term has coefficient one in the comparison. A checkpoint comparison transfers this mixed guarantee to the stream (\cref{sec:curvature}). Since $\ell(O)\ge(1-\vartheta)f(O)$, this gives $\rho_\vartheta=1-(\sqrt2-1)\vartheta$, up to $\varepsilon$, with $O(\varepsilon^{-2})$ recourse. The threshold ranges from $1$ at $\vartheta=0$ to $\bconst$ at $\vartheta=1$. Modular augmentation of the hard instance preserves exact hiding and matches this bound for every fixed $\vartheta>0$.

\section{The General Upper Bound}\label{sec:cores}

This section proves \cref{thm:main}\textup{(i)}. In \cref{sec:core-construction}, classical greedy certificates and submodular sampling yield a polynomial-size anchored-action LP using only current values. In \cref{sec:core-threshold}, LP duality and finite-grid supporting-line geometry show that this distribution achieves coefficient $\bconst=2-\sqrt2$ in expectation for every fixed legal future. The checkpoint conversion of \cref{lem:checkpoint} then yields a $(\bconst-\varepsilon)$ approximation with $O(\varepsilon^{-2})$ worst-case symmetric recourse. Fix a current set $X$ and capacity $\kappa\ge1$ throughout.

\subsection{Anchored Sampling and Current-Only Mixing}\label{sec:core-construction}

Our construction begins with two greedy certificates: one captures the value already obtained, while the other bounds the remaining contribution of an optimal comparison set. If $|X|\le\kappa$, return $X$, since $h(X)=P_h$. Otherwise, run greedy \citep{NWF78} for $2\kappa+1$ positions, padding the sequence with distinct nulls after exhausting $X$. Let $G_j\subseteq X$ denote the real elements among the first $j$ positions. Null positions require no function evaluation. Write $v_j=f(G_j)$, and define $D_j=\kappa(v_{j+1}-v_j)$ for $0\le j\le2\kappa$. For a fixed legal future $h$, let $Y_j=h(G_j)$. The two greedy certificates are
\begin{equation}\label{eq:endpoints}
Y_j\ge v_j,\qquad Y_j\ge P_h-D_j.
\end{equation}
The value certificate uses $h\ge f$. For $O$ attaining $P_h$, the classical greedy residual bound \citep{NWF78} transfers through marginal domination in \eqref{eq:legal}: monotonicity and submodularity of $h$ give
\[
P_h\le Y_j+\sum_{i\in O\setminus G_j}h(i\mid G_j)
\le Y_j+\sum_{i\in O\setminus G_j}f(i\mid G_j)\le Y_j+D_j,
\]
where the last inequality uses the greedy maximum marginal and $|O|\le\kappa$.
The current chain determines $v_j,D_j$, while $Y_j,P_h$ remain unknown. Longer prefixes strengthen both certificates but may exceed the cardinality constraint. For $0\le a<\kappa<b\le2\kappa$, the anchored sampler $A_{a,b}$ keeps $G_a$ and uniformly selects $\kappa-a$ of the positions $a+1,\ldots,b$, discarding any selected nulls. Thus $A_{a,b}\subseteq X$ and $|A_{a,b}|\le\kappa$. Let $\lambda=(\kappa-a)/(b-a)$.

For a submodular function $\chi:2^N\to\mathbb R$ and a random
subset $Z\subseteq N$ with $\Pr(i\in Z)=\lambda\in[0,1]$ for every
$i\in N$, \citep[Lemma~2.2]{FMV11} gives
\begin{equation}\label{eq:equal-marginal-sampling}
\E[\chi(Z)]\ge(1-\lambda)\chi(\varnothing)+\lambda\chi(N).
\end{equation}
Neither independence nor normalization is required, allowing sampling
without replacement with a possibly nonzero anchor value.

\begin{lemma}[Anchored interpolation]\label{lem:anchor}
For every fixed legal future $h$,
\[
\E h(A_{a,b})\ge(1-\lambda)Y_a+\lambda Y_b.
\]
\end{lemma}
\begin{proof}
Let $I$ be the uniform $(\kappa-a)$-subset of $N=\{a+1,\ldots,b\}$
drawn by the sampler. For $J\subseteq N$, let $\operatorname{real}(J)$
denote the real elements at positions in $J$, ignoring nulls, and set
$\chi(J)=h(G_a\cup\operatorname{real}(J))$. Fixing $G_a$ preserves
submodularity of $h$, as does relabeling distinct real elements by their
positions and adjoining zero-marginal nulls. Since
$A_{a,b}=G_a\cup\operatorname{real}(I)$, $\chi(\varnothing)=Y_a$,
$\chi(N)=Y_b$, and $\Pr(i\in I)=\lambda$ for every $i\in N$, we have
$\E h(A_{a,b})=\E[\chi(I)]$, and \eqref{eq:equal-marginal-sampling}
proves the claim.
\end{proof}
The same sampling principle is used in \citep[Lemma~6.3(iii)]{DFL+25}
for uniform subsampling of a threshold-augmented greedy solution.
Varying both endpoints realizes every chord strictly crossing $\kappa$.
This completeness lets the supporting-line argument in
\cref{sec:core-threshold} control certificates along the entire greedy chain.
Let $(\mu_{j,\sigma},\nu_{j,\sigma})$ denote the slope and intercept
in $P_h$ of the endpoint certificate $Y_j\ge\mu_{j,\sigma}P_h+\nu_{j,\sigma}$,
where
$\sigma\in\{\mathsf v,\mathsf r\}$ and
\[
(\mu_{j,\mathsf v},\nu_{j,\mathsf v})=(0,v_j),
\qquad
(\mu_{j,\mathsf r},\nu_{j,\mathsf r})=(1,-D_j).
\]
A \emph{labeled action} $\pi=(a,b,\sigma_a,\sigma_b)$ selects one
certificate at each endpoint without changing the sampler
$A_\pi=A_{a,b}$, and has coefficients
\begin{equation}\label{eq:chord-coefficients}
(\mu_\pi,\nu_\pi)
=(1-\lambda)(\mu_{a,\sigma_a},\nu_{a,\sigma_a})
+\lambda(\mu_{b,\sigma_b},\nu_{b,\sigma_b}),
\qquad \lambda=\frac{\kappa-a}{b-a}.
\end{equation}
By \cref{lem:anchor} and \eqref{eq:endpoints}, each action satisfies
$\E h(A_\pi)\ge\mu_\pi P_h+\nu_\pi$.
Each of the $\kappa^2$ endpoint pairs gives four certificate combinations
for the same sampler, and the two deterministic labels $(0,v_\kappa)$
and $(1,-D_\kappa)$ for output $G_\kappa$ give $4\kappa^2+2$ labeled actions
in total.

For a distribution $p$ over labeled actions, draw $\pi\sim p$ and
return $A=A_\pi$. The coefficients depend only on current values,
and averaging as in \eqref{eq:ideas-mixture} gives
$\E h(A)\ge(\sum_\pi p_\pi\mu_\pi)P_h$ for every fixed legal future $h$
whenever $\sum_\pi p_\pi\nu_\pi\ge0$. Thus the current-only LP maximizes
this certified coefficient:
\begin{equation}\label{eq:core-lp}
\Gamma=\max\left\{\sum_\pi p_\pi \mu_\pi:
\sum_\pi p_\pi \nu_\pi\ge0,\quad
\sum_\pi p_\pi=1,\quad p_\pi\ge0\right\}.
\end{equation}
The same optimal mixture, chosen without knowledge of $h$ or $P_h$,
therefore satisfies $\E h(A)\ge\Gamma P_h$ for every fixed legal future $h$.

\begin{lemma}[Duality and two-action support]\label{lem:core-dual}
The linear program \eqref{eq:core-lp} is feasible and has value $\Gamma=\inf_{\omega\ge0}\max_\pi(\mu_\pi+\omega \nu_\pi)$. It admits an optimal mixture supported on at most two labeled actions. Given the greedy chain, this mixture can be computed in $O(\kappa^2\log(\kappa+1))$ arithmetic operations and comparisons.
\end{lemma}
\begin{proof}
Concentrating on the value certificate at $G_\kappa$ is feasible since
$v_\kappa\ge0$. The feasible set is compact, so an optimum exists,
and finite LP duality gives the formula with multiplier $\omega\ge0$
for the intercept constraint. A distribution $p$ gives
$(\sum_\pi p_\pi\nu_\pi,\sum_\pi p_\pi\mu_\pi)$, and these pairs
form the convex hull of the action points $(\nu_\pi,\mu_\pi)$.
An optimum maximizing $\mu$ over $\nu\ge0$ can be chosen at a feasible
hull vertex or where a hull edge meets $\nu=0$, so at most two labeled
actions suffice. Enumerating the $4\kappa^2+2$ points takes $O(\kappa^2)$
operations, and sorting by $\nu$ and building the upper hull takes
$O(\kappa^2\log(\kappa+1))$ arithmetic operations and comparisons.
A linear scan of feasible vertices and boundary intersections finds
an optimum, and its action labels and interpolation weights recover
the mixture.
\end{proof}

For $|X|>\kappa$, the robust core draws a labeled action $\pi$ with probabilities $(p_\pi)_\pi$ from an optimal solution of \eqref{eq:core-lp} and returns $A_\pi$.

\Needspace{4\baselineskip}
\subsection{The Supporting-Line Threshold}\label{sec:core-threshold}

By \cref{lem:core-dual}, it suffices to show that every dual value is at least $\bconst$. We make the geometry of \cref{sec:ideas-threshold} exact for every finite $\kappa$ by piecewise-linear interpolation, without taking a limit.

\begin{theorem}[A $(2-\sqrt2)$-robust core]\label{thm:anchored}
For every current set $X$ and integer capacity $\kappa\ge1$, there is a randomized algorithm that uses only current value queries and outputs a set $A\subseteq X$ with $|A|\le\kappa$ on every outcome. For every fixed legal future $h$, the output distribution satisfies
\[
\E h(A)\ge(2-\sqrt2)\max_{O\subseteq X,\,|O|\le\kappa}h(O).
\]
The algorithm makes $O(|X|\kappa)$ current value queries and runs in polynomial time.
\end{theorem}
\begin{proof}
The algorithm returns $X$ if $|X|\le\kappa$. Otherwise, it samples a labeled action from an optimal solution of the linear program \eqref{eq:core-lp} and returns the corresponding set.

\emph{The supporting line.} Fix $\omega\ge0$ and set $q_j=j/\kappa$ and $u_j=\max\{\omega v_j,1-\omega D_j\}$. For $\omega>0$, $u_j$ is the better normalized endpoint certificate at the candidate benchmark $P=1/\omega$. Let $\gamma=\gamma(\omega)=\max_\pi(\mu_\pi+\omega\nu_\pi)$. Maximizing over certificate labels in \eqref{eq:chord-coefficients} gives all two-point chords strictly crossing $q=1$: their endpoints satisfy $q_a<1<q_b$ and their weights are $1-\lambda,\lambda$, with $\lambda=(1-q_a)/(q_b-q_a)$. Together with the deterministic singleton at $q_\kappa=1$, these are exactly the action family after maximizing over labels. Their maximum height is the value at $1$ of the least concave majorant of the points $(q_j,u_j)$.

Since $v_j$ is nondecreasing and $D_j$ is nonincreasing, $u_j$ and its least concave majorant are nondecreasing. At the interior point $1$, the latter has a supporting line of nonnegative slope. Its intercept satisfies $c\ge u_0\ge0$, so
\begin{equation}\label{eq:line}
u_j\le c+dq_j,\qquad c+d=\gamma,\qquad c,d\ge0.
\end{equation}

\emph{Forced growth.} If $\gamma\ge3/5>\bconst$, the desired bound holds. Otherwise, linearly interpolate the scaled greedy values $s(q_j)=\omega v_j$. Since $v_{j+1}-v_j=D_j/\kappa$, the value and residual terms in \eqref{eq:line} respectively bound the interpolant from above and its derivative from below:
\begin{equation}\label{eq:greedy-differential}
s(0)=0,\qquad s(t)\le c+dt,\qquad s'(t)\ge1-c-dt\quad\text{a.e. on }[0,2].
\end{equation}
Indeed, on the $j$th cell $s'(t)=\omega D_j\ge1-c-dq_j\ge1-c-dt$.

Integrating the derivative bound over $[0,2]$ and using $s(0)=0$ gives
\[
2-2\gamma=2(1-c)-2d\le s(2)\le c+2d=\gamma+d.
\]
Thus $d\ge2-3\gamma>0$. Since $\gamma<3/5$, the critical point $t_*=(1-\gamma)/d$ satisfies $0<t_*\le(1-\gamma)/(2-3\gamma)<2$, so \eqref{eq:greedy-differential} applies throughout $[0,t_*]$. Integrating to $t_*$ and comparing with $s(t_*)\le c+dt_*$ gives $c\ge(1-\gamma)t_*-dt_*^2/2=(1-\gamma)^2/(2d)$. Hence
\begin{equation}\label{eq:core-quadratic}
(1-\gamma)^2\le2cd\le\frac{(c+d)^2}{2}=\frac{\gamma^2}{2},
\end{equation}
which yields $\gamma\ge2-\sqrt2$. Continuous equality fixes $c=d=\bconst/2$, $t_*=\sqrt2$, and $s(t)=(1-c)t-dt^2/2$ on $[0,\sqrt2]$, but not the continuation (\cref{fig:profiles}).

By \cref{lem:core-dual}, $\Gamma=\inf_{\omega\ge0}\gamma(\omega)\ge\bconst$. Thus an optimal mixture satisfies $\E h(A)\ge\Gamma P_h\ge\bconst P_h$ for every fixed legal future. Greedy uses $O(|X|\kappa)$ current value queries, and \cref{lem:core-dual} bounds the remaining computation.
\end{proof}

In our notation, \citep[Lemma~4.6, Eq.~(4)]{MZ15} also gives $(1-\gamma)^2\le2cd$ with extremal values $c=d=\bconst/2$. Here, the anchored-action LP and current greedy certificates yield this scalar tradeoff \eqref{eq:core-quadratic} through a finite-grid supporting line, establishing the robust-core guarantee.

\begin{proof}[Proof of the upper bound in \cref{thm:main}]
Apply \cref{lem:checkpoint} to the robust core of \cref{thm:anchored}, with $B=\lceil3/\varepsilon\rceil\ge4$. Its coefficient is at least $\bconst(1-2/B)(1-1/B)\ge\bconst-3\bconst/B\ge\bconst-\bconst\varepsilon>\bconst-\varepsilon$. The worst-case recourse is at most $8\lceil3/\varepsilon\rceil^2+2$. Each arrival invokes at most one core computation or greedy run, using $O(|X_t|k)$ current value queries.
\end{proof}

\section{A Matching Query--Recourse Lower Bound}\label{sec:lower}

We follow the final-arrival template of \citep[Section~5.2]{DFL+25}. Our construction additionally hides a useful $k$-set among $(m+1)k$ old elements from current queries. A single monotone submodular function and fixed arrival order enforce both hiding and costly repair. We state the finite bound before constructing this instance.

\begin{theorem}[Finite query--recourse lower bound]\label{thm:finite-lower}
Fix $T\in[\nicefrac{7}{5},\nicefrac{3}{2}]$ and integers $m\ge32$, $k\ge10$. Let $n=(m+1)k$ and $\mathcal R(T)=(T^2+2)/(T^2+2T+2)$. Consider any randomized value-oracle algorithm maintaining at most $k$ arrived elements and terminating almost surely on every model instance. Assume at most $Q$ queries before the final arrival and at most $C$ symmetric changes there, on every model instance with $n+1$ elements and every random outcome. For some $k$-set $A\subseteq X=[n]$ and normalized monotone submodular $f_A$ on $X\cup\{r\}$, arrival order $1,2,\ldots,n,r$ yields
\begin{equation}\label{eq:finite-lower-ratio}
\E f_A(S_{n+1})
\le \left(\frac{\mathcal R(T)+38/m+5C/k}{1-5/k}
       +2(Q+1)e^{-k/(8192m^2)}\right)\OPT_k(X\cup\{r\}).
\end{equation}
The unique minimizer is $T=\sqrt2$, with $\mathcal R(\sqrt2)=2-\sqrt2$. For every fixed $\zeta>0$, some integer $m\ge32$ and constants $c_\zeta,d_\zeta>0$ ensure that, for all sufficiently large $k$, any such algorithm whose bounds $Q\le e^{c_\zeta k}$ and $C\le c_\zeta k$ hold on every model instance of at most $d_\zeta k$ elements and every random outcome has a fixed instance with $T=\sqrt2$ and $n+1\le d_\zeta k$ satisfying
\[
\E f_A(S_{n+1})<(2-\sqrt2+\zeta)\OPT_k(X\cup\{r\}).
\]
Instances are fixed independently of the algorithm's randomness, over which we take expectations. Queries may adaptively use any subset of arrived elements. Queries and computation after the final arrival are unrestricted. Old elements have marginals at most $5/k$ on all backgrounds, with or without $r$.
\end{theorem}

The lower bound is unconditional and information-theoretic. The hidden partition specifies an explicit objective. Polynomially many queries before the final arrival and constant recourse cannot guarantee any fixed improvement over $2-\sqrt2$, even with unrestricted computation throughout and unrestricted queries afterward.

\paragraph{How the bound follows from the construction.}
\Cref{lem:hard-profiles} gives a construction with exact hiding, final-arrival compatibility, and small marginals against every background set. \Cref{lem:transcript} transfers concentration for fixed sets to the algorithm's adaptive transcript and pre-arrival output. The final proof bounds the old solution's value after adjoining $r$ for free, adds the gain permitted by recourse, and compares with a feasible hidden set. Averaging, including the failure event, then fixes a hard instance independently of the algorithm's random choices.

\subsection{Coupled Current and Future Profiles}

We adapt the normalized continuous initial segment of the hidden-color construction in \citep[Lemma~2.1 and Definition~2.3]{HKMY22}. Appendix~\ref{app:hkmy-normalization} records its relation to the finite original function. Our continuation and subsequent modifications are verified separately below. Fix $T\in[\nicefrac{7}{5},\nicefrac{3}{2}]$ and an integer $m\ge32$, and define
\begin{equation}\label{eq:vw}
\begin{gathered}
d=(T^2/2+T+1)^{-1},\qquad c=dT^2/2,\\
v(t)=\begin{cases}d((T+1)t-t^2/2),&0\le t\le T,\\
1-de^{-(t-T)},&t\ge T,\end{cases}\qquad
w(t)=\begin{cases}c+dt,&0\le t\le T,\\v(t),&t\ge T.\end{cases}
\end{gathered}
\end{equation}
The relations $1-c=d(T+1)$ and $c=dT^2/2$ give tangential contact at $T$. Both profiles are nondecreasing and concave. We choose the shared exponential continuation to make $v$ of class $C^2$ and $w$ of class $C^1$ with locally Lipschitz derivative, and to control the derivative and correction budgets on the full count rectangle. The profiles satisfy
\begin{equation}\label{eq:profile-identities}
w(t)-v(t)=\frac d2(T-t)_+^2,\qquad c+v'(0)=1.
\end{equation}
The limiting balanced value is $w(1)=d+c=\mathcal R(T)$, while the hidden-complement value is $c+v'(0)=1$. Since $\mathcal R(T)-\bconst=(\sqrt2-1)(T-\sqrt2)^2/(T^2+2T+2)$, the minimum occurs at $T=\sqrt2$. At this parameter, the inherited initial segment has $c=d=\bconst/2$ and satisfies the continuous equality conditions of \eqref{eq:core-quadratic} (\cref{fig:profiles}). These conditions do not determine the tail. The finite construction only approaches the limiting ratio after the modifications below.

Partition $X=A\sqcup\bar A$ with $|A|=k$ and $|\bar A|=mk$. For any queried set $S\subseteq X$, define $x=|S\cap A|/k$ and $y=|S\cap\bar A|/k$, so $(x,y)\in[0,1]\times[0,m]$. Then $s=x+y$ is normalized query size and $u=x-y/m$ measures imbalance. On $x+y=1$, the balanced point is $(1/(m+1),m/(m+1))$, whereas $A$ corresponds to $(1,0)$. The final proof uses a Lipschitz bound to compare actual sets with the possibly off-grid balanced point.

\subsection{Exact Hiding and Adaptive Queries}

We adapt the exact-flattening principle of \citep[Lemma~3.2]{Von13}, whose symmetry-gap framework extends the oracle lower bounds of \citep{FMV11}. Here exact hiding must coexist with a further requirement: the current and future profiles must define the two branches of a single monotone submodular objective on all subsets. Besides monotonicity and diminishing coordinate gradients within each branch, we enforce $K-F\ge0$ and $\nabla K\le\nabla F$. These inequalities make the final element's marginal nonnegative and nonincreasing as old elements are added.

\emph{The coupled lift and regularization.} Set $\omega_m=1+1/m$, $t=\omega_m{}y$, and $\phi_m(s)=32(1-e^{-s})/m$. Lift the scalar curves to
\begin{equation}\label{eq:FK}
F(x,y)=v(t)+uv'(t)+\phi_m(s),\qquad
K(x,y)=w(t)+uv'(t)+\phi_m(s).
\end{equation}
The shared tangent term gives $K-F=w(t)-v(t)$, which is nonnegative and coordinatewise nonincreasing by \cref{eq:profile-identities}. Since this difference depends only on $y$, we have $\partial_xK=\partial_xF$: before flattening, an element of the hidden set has the same marginal with or without $r$. This cross-branch comparison therefore has no slack. At the balanced unit-size point, $K$ has value $\mathcal R(T)+\phi_m(1)$, whereas $K(1,0)=1+\phi_m(1)$.

The common regularizer first ensures diminishing returns within each branch. Before contact, the unregularized future profile $K_0=w(t)+uv'(t)$ has $(K_0)_{yy}=2\omega_m{}d/m>0$, violating the required decrease of marginals between old elements of the same type. Adding $\phi_m$ contributes $32e^{-s}/m$ to each gradient and $-32e^{-s}/m$ to each Hessian entry, restoring the signs and leaving quantitative margins without changing $K-F$ (Appendix~\ref{app:profile-derivatives}). This use of a concave function of total mass follows the regularization principle in the proof of \citep[Lemma~3.2]{Von13}, which adds a common correction to two comparison functions. We use an explicit exponential regularizer on the full count rectangle.

\emph{Exact hiding and the future correction.} Let $\delta=1/(64m)$ and define $\clip(u,-\delta,\delta)=\min\{\delta,\max\{-\delta,u\}\}$. Move the current tangent point toward $s=t+u$, and correct the future branch as follows:
\begin{equation}\label{eq:clamped}
\begin{aligned}
z&=t+\clip(u,-\delta,\delta),\\
F_\delta(x,y)&=v(z)+(s-z)v'(z)+\phi_m(s),\\
K_\delta(x,y)&=K(x,y)+32\delta e^{-s}.
\end{aligned}
\end{equation}
Inside the band, the tangent point reaches $s$, so
\begin{equation}\label{eq:exact-flat-band}
|u|\le\delta\quad\Longrightarrow\quad z=s
\quad\Longrightarrow\quad F_\delta(x,y)=v(s)+\phi_m(s).
\end{equation}
Within this band, current answers depend exactly on query size. Moving a concave function's tangent point toward $s$ decreases its tangent value at $s$, so $F_\delta\le F$ and hence $K_\delta\ge K\ge F\ge F_\delta$.

This value ordering does not guarantee the required gradient ordering. On the quadratic piece with $0<u<\delta$ and $s=t+u<T$, flattening the current branch without correcting $K$ gives
\[
\partial_xF_\delta-\partial_xK=v'(s)-v'(t)=-du<0.
\]
The comparison that previously held with equality is reversed. Before $r$ arrives, every allowed query evaluates only $F_\delta$. Correcting $K$ therefore leaves the current flat band unchanged. The future correction in \cref{eq:clamped} restores compatibility: Appendix~\ref{app:profile-derivatives} gives
\begin{equation}\label{eq:gradient-perturbation}
\|\nabla F_\delta-\nabla F\|_\infty\le16\delta e^{-s},
\end{equation}
and for $i\in\{x,y\}$,
\begin{equation}\label{eq:compatible-correction}
\begin{aligned}
\partial_iF_\delta-\partial_iK_\delta
&=(\partial_iF_\delta-\partial_iF)
 +(\partial_iF-\partial_iK)+32\delta e^{-s}\\
&\ge16\delta e^{-s}.
\end{aligned}
\end{equation}
Consequently, the final element's marginal is nonincreasing in both count coordinates.

The correction also fits within the derivative margins. Before this correction, the regularized profiles $F_\delta$ and $K$ have first derivatives at least $22e^{-s}/m$ and Hessian entries at most $-23e^{-s}/m$ wherever the latter exist, as proved in Appendix~\ref{app:profile-derivatives}. The correction lowers each future gradient and raises each future Hessian entry by only $32\delta e^{-s}=e^{-s}/(2m)$. Both branches retain the margins $16e^{-s}/m$ used in \cref{lem:flat}. Thus our choice of a band of width $\Theta(1/m)$ preserves the derivative signs as well as the cross-branch ordering. At total mass one, the common regularizer and future correction add the same amount at the hidden and balanced points, preserving their absolute value difference $1-\mathcal R(T)$.

\emph{The discrete instance.} Let $X=[n]$, where $n=(m+1)k$. For a $k$-set $A\subseteq X$, let $\bar A=X\setminus A$ and, for $S\subseteq X$, let $u_A(S)=(|S\cap A|-|S\cap\bar A|/m)/k$ denote its imbalance. Define one function on $X\cup\{r\}$ by
\begin{equation}\label{eq:hidden-oracle}
f_A(S)=F_\delta\left(\frac{|S\cap A|}{k},\frac{|S\cap \bar A|}{k}\right),
\quad
f_A(S\cup\{r\})=K_\delta\left(\frac{|S\cap A|}{k},\frac{|S\cap \bar A|}{k}\right).
\end{equation}
The arrival order is $1,2,\ldots,n,r$, independently of $A$. Let $\psi(s)=v(s)+\phi_m(s)$ for $s\in[0,m+1]$. The value $K_\delta(1,0)=f_A(A\cup\{r\})$ concerns $k+1$ elements, whereas the balanced unit-size point may lie off the count grid.

\begin{lemma}[Exact hiding with a legal future]\label{lem:hard-profiles}
For the parameters above and every $k$-set $A\subseteq X$, the construction has the following properties.
\begin{enumerate}[label=\textup{(\roman*)},leftmargin=2em]
\item \emph{Exact hiding.} The function $\psi$ is nondecreasing and concave, with $\psi(0)=0$. For every $S\subseteq X$ with $|u_A(S)|\le\delta$, we have $f_A(S)=\psi(|S|/k)$.
\item \emph{Final-arrival compatibility.} The function $f_A$ is normalized, monotone, and submodular on all of $2^{X\cup\{r\}}$, including sets larger than $k$.
\item \emph{Small marginals.} Every coordinate derivative of $F_\delta$ and $K_\delta$ lies in $[0,5]$. For every old element $e\in X$ and every background $U\subseteq(X\cup\{r\})\setminus\{e\}$, we have $0\le f_A(U\cup\{e\})-f_A(U)\le5/k$.
\item \emph{Value gap.} The future profile satisfies
\[
\begin{aligned}
K_\delta\left(\frac1{m+1},\frac m{m+1}\right)
 &=\mathcal R(T)+\phi_m(1)+32\delta e^{-1}\\
 &\le\mathcal R(T)+\frac{32}{m}+32\delta,\\
K_\delta(1,0)&=1+\phi_m(1)+32\delta e^{-1}\ge1.
\end{aligned}
\]
\end{enumerate}
\end{lemma}
\begin{proof}
Part~(i) follows from \cref{eq:exact-flat-band} and the properties of $v$ and $\phi_m$, and part~(iv) follows by substitution. For parts~(ii)--(iii), \cref{lem:flat} proves the derivative bounds and the coordinatewise monotonicity of the gradients on the full rectangle. An old-element marginal is the integral of a coordinate derivative over an interval of length $1/k$. It therefore lies in $[0,5/k]$ and decreases as either old-element count increases, including when two distinct elements have the same type. The inequalities $K_\delta-F_\delta\ge0$ and $\nabla K_\delta\le\nabla F_\delta$ make the marginal of $r$ nonnegative and nonincreasing. Together with $F_\delta(0,0)=0$, these facts give normalization, monotonicity, and submodularity. The continuous-to-discrete verification, including intervals crossing piece boundaries, is given in \cref{lem:continuous-to-discrete}, adapting \citep[Lemma~3.1]{Von13}.
\end{proof}

\emph{A legal reference and arbitrary current queries.} Choose $A$ uniformly among the $k$-subsets of $X$. To handle adaptive queries, we realize the size-dependent current answers as a valid reference instance on the same ground set:
\begin{equation}\label{eq:reference-oracle}
f_{\rm ref}(U)=\psi\!\left(\frac{|U\cap X|}{k}\right).
\end{equation}
By part~(i) of \cref{lem:hard-profiles}, this function is normalized, monotone, and submodular, and $r$ has zero marginal. Its validity lets us apply the algorithm's uniform query bound and almost-sure termination requirement.\footnote{A query cap alone does not exclude infinite computation between queries.} We compare transcripts only before $r$ arrives, when these reference answers depend only on query size and are independent of $A$.

For a fixed current query $S\subseteq X$, $Z=|S\cap A|$ is hypergeometric with mean $|S|/(m+1)$, and $u_A(S)=(m+1)(Z-\E Z)/(mk)$. Viewing $A$ as a sample of $k$ identifiers, Hoeffding's without-replacement inequality \citep[Section~6]{Hoe63} gives
\begin{equation}\label{eq:hypergeo}
\Prob_A\{|u_A(S)|>\delta\}\le2e^{-\delta^2k/2}.
\end{equation}
The sample size is $k$, so the estimate holds for every fixed query regardless of $|S|$, including queries larger than $k$. This covers queries to any subset of the arrived history. With $\delta=1/(64m)$, the exponent is $k/(8192m^2)$, recording the concentration cost of our band width. Outside the exceptional event, the answer is exactly $\psi(|S|/k)$, independently of $A$. The next lemma applies this fixed-set estimate to the queries and output of a reference execution to obtain the adaptive-transcript guarantee.

\begin{lemma}[Transcript hiding]\label{lem:transcript}
Suppose an algorithm terminates almost surely on every instance satisfying the model and makes at most $Q$ current value queries before the last arrival on each such instance with $n+1$ elements. With probability at least $1-2(Q+1)e^{-\delta^2k/2}$ over $A$ and the algorithm's randomness, its answers and output immediately before $r$ agree with a reference-oracle execution against $\psi$, and that output has imbalance at most $\delta$.
\end{lemma}
\begin{proof}
We adapt the reference-transcript argument from the proof of \citep[Lemma~3.3]{Von13}. The query bound and almost-sure termination apply to the legal reference instance \cref{eq:reference-oracle}. Fix its entire random tape outside the null set of nonterminating executions. Its at most $Q$ queries and pre-arrival output are then fixed independently of $A$. Apply the size-independent bound \cref{eq:hypergeo} and a union bound to these $Q+1$ sets, including the possibly unqueried output. Until the first different answer, the real execution has the same state and next query. On the good event every reference answer agrees, so no first difference occurs, stopping decisions and outputs agree, and the output is balanced. Averaging over the random tape proves the claim.
\end{proof}

\subsection{Proof of the Finite Query--Recourse Bound}

\begin{proof}[Proof of \cref{thm:finite-lower}]
We follow the final-arrival repair template of \citep[Section~5.2, Theorem~5.6]{DFL+25}, using the marginal bounds of our construction. Let $S^-=S_n$ and $S^+=S_{n+1}$ be the outputs before and after the critical arrival, and define $\mathcal E$ by agreement of the current answers and $S^-$ with the reference execution and $|u_A(S^-)|\le\delta$. By \cref{lem:transcript}, $\Prob(\mathcal E^c)\le2(Q+1)e^{-\delta^2k/2}$.

\emph{Before repair.} On $\mathcal E$, feasibility of $S^-$ bounds its normalized hidden mass by $x=|S^-\cap A|/k\le1/(m+1)+\delta$. Add elements of $\bar A$ until the entire current set has size $k$, which is possible because $|\bar A|=mk$ and cannot decrease its future value. By \cref{lem:hard-profiles}(iii), both coordinate gradients of the future profile lie in $[0,5]$. Its derivative along $x+y=1$ is their difference, so its value is $5$-Lipschitz as a function of $x$. At its balanced point $x_*=1/(m+1)$, we have $u=0,t=1$. By monotonicity, the value of the completion together with $r$ upper-bounds $f_A(S^-\cup\{r\})$, with $k+1$ elements allowed in this comparison. Therefore
\begin{equation}\label{eq:balanced-payoff}
\begin{aligned}
f_A(S^-\cup\{r\})
&\le \mathcal R(T)+\frac{32}{m}+32\delta
          +5\left(\frac1{m+1}+\delta\right)\\
&\le \mathcal R(T)+\frac{38}{m}.
\end{aligned}
\end{equation}
The estimate uses $0<e^{-s}\le1$ and $w(1)=\mathcal R(T)$, since $1<T$.

\emph{Gain from repair.} Let $I=(S^+\cap X)\setminus S^-$ be the newly installed old elements, so $|I|\le|S^+\triangle S^-|\le C$. By \cref{lem:hard-profiles}(iii), every old-element marginal is at most $5/k$ against any background. By monotonicity and $S^+\subseteq S^-\cup I\cup\{r\}$, adjoining $r$ without charging recourse bounds the gain from repair by $5|I|/k\le5C/k$. This estimate allows unrestricted post-arrival queries and computation.

\emph{A feasible comparison set.} At the hidden point $(1,0)$, the future profile has the exact value
$K_\delta(1,0)=c+v'(0)+\phi_m(1)+32\delta e^{-1}\ge1$.
Removing one current element costs at most $5/k$ by \cref{lem:hard-profiles}(iii). Thus, for any $a_0\in A$,
\begin{equation}\label{eq:comparator}
f_A\bigl(\{r\}\cup(A\setminus\{a_0\})\bigr)\ge1-5/k.
\end{equation}
This set has exactly $k$ elements. On $\mathcal E^c$, any feasible output has ratio at most one. The optimum is the same for every $A$ because the instances differ only by a permutation of current identifiers. Averaging the ratio over $A$ and the random tape therefore gives \cref{eq:finite-lower-ratio} for some fixed $A$. The fraction in \eqref{eq:finite-lower-ratio} is at most $\mathcal R(T)+76/m+10(C+1)/k$, since $1/(1-5/k)\le2$ and $\mathcal R(T)\le1$.
\end{proof}

To deduce the asymptotic conclusion and \cref{thm:main}(ii), fix $\zeta>0$ and let $\zeta_0=\min\{\zeta,1\}$. Take $T=\sqrt2$, so $\mathcal R(T)=\bconst$, and choose an integer $m\ge32$ with $76/m<\zeta_0/4$. These choices fix the profiles before $k$ is chosen. Set $a=1/(8192m^2)$ and $c_\zeta=\min\{a/2,\zeta_0/80\}$. For all sufficiently large $k$, the bounds $C\le c_\zeta k$ and $Q\le e^{c_\zeta k}$ imply
\[
\frac{10(C+1)}k\le\frac{\zeta_0}4,
\qquad
2(Q+1)e^{-ak}\le4e^{-ak/2}\le\frac{\zeta_0}4.
\]
The finite bound is therefore strictly below $\bconst+\zeta$. Taking $d_\zeta=m+2$ bounds the instance size $(m+1)k+1$ by $d_\zeta k$. The instance may depend on the algorithm, but its objective and arrival order are fixed before the algorithm's random choices.

\section{Prices for Unknown Futures}\label{sec:prices}

The anchored core gives a multiplicative guarantee for every legal future. For the structural results, we seek a larger coefficient for coverage and a guarantee that preserves a modular component separately under bounded curvature. We use prices computed from current values to express these comparisons through a linear term. \Cref{lem:price-gap} turns a small price gap into a value guarantee against every legal future, and \cref{lem:scale} constructs the required prices from a current potential. We use this interface in the coverage and modular-preservation arguments in \cref{sec:structure,sec:curvature}.

\subsection{The Price-Gap Principle}

Fix a normalized current function $f$ on $X=[n]$, with $n\ge1$, and define $M_i=f(\{i\})$ and $M=\max_iM_i$. Let $\calH_f$ be the legal futures in \cref{eq:legal}. For a nonnegative vector $\boldsymbol x$ of sampling parameters, define the Poisson extension $H_h(\boldsymbol x)=\E h(Z_{\boldsymbol x})$, where coordinate $i$ is present independently with probability $1-e^{-x_i}$, as in Poisson rounding \citep[Section~4.1, Algorithm~2]{DRY11}. Set $\aconst=1-e^{-1}$ and $P_h=\max_{O\subseteq X,\ |O|\le\kappa}h(O)$. For any finite current set $Y$, define
\[
\mathcal P_\kappa(Y)=\Bigl\{\boldsymbol x\in[0,1]^Y:\sum_{i\in Y}x_i\le\kappa\Bigr\},
\qquad \mathcal P_\kappa=\mathcal P_\kappa(X),\qquad\kappa\ge1.
\]
The Poisson reparameterization also appears in \citep[Lemmas~2.5 and~3.1]{BNW25}. The sampled set $Z_{\boldsymbol{x}}$ may exceed capacity even when $\boldsymbol{x}\in \mathcal P_\kappa$. Rounding uses only current data to transfer its expected-value guarantee to a feasible sample.

Fixed sets may overlap the current set, as the following lemma shows.

\begin{lemma}[Fixed sets and overlapping coordinates]\label{lem:overlap}
Let $f$ be monotone submodular on a ground set containing $X\cup R$, where $R$ is any fixed set, possibly intersecting $X$. Then $h_R(S)=f(S\cup R)$ is a legal future on $X$, and its Poisson extension $H_R=H_{h_R}$ depends only on coordinates in $X\setminus R$.
\end{lemma}
\begin{proof}
Because $R$ is fixed before sampling, adjoining it preserves monotonicity and submodularity and gives $h_R\ge f$. For $i\notin R$, diminishing returns bounds its marginal at $S\cup R$ by its marginal at $S$, proving \cref{eq:legal}. For $i\in R$, the marginal is zero, so its random inclusion has no effect on $H_R$.
\end{proof}

Prices assign nonnegative weights to current elements, using only current values, to give a linear lower bound on the difference between expected and scaled comparison values. For $0\le\alpha\le1$, $\xi\ge0$, and $\boldsymbol x\in \mathcal P_\kappa$, an $(\alpha,\xi)$-price is a vector $\boldsymbol p\in \Pi_f:=\prod_i[0,M_i]$ such that
\begin{equation}\label{eq:operational-price}
H_h(\boldsymbol x)-\alpha h(O)\ge\ip{\boldsymbol p}{\boldsymbol x-\one_O}-\xi M
\quad\text{for every }h\in\calH_f,\ O\subseteq X,\ |O|\le\kappa.
\end{equation}
Both $\boldsymbol x$ and $\boldsymbol p$ are chosen before the future is known. A universal $\alpha$-price has $\xi=0$, satisfies the inequality for every $O\subseteq X$ independently of capacity, and obeys the coordinate bounds $p_i\le M_i$.

The price inequality reduces every future comparison to the same current linear objective. We therefore seek a point whose own prices admit little improvement over $\mathcal P_\kappa$, with the largest such improvement measured by the price gap below.

\begin{lemma}[A price gap certifies the response]\label{lem:price-gap}
Suppose $\boldsymbol p$ satisfies \cref{eq:operational-price} at $\boldsymbol x\in \mathcal P_\kappa$, and define its price gap $\gap(\boldsymbol p,\boldsymbol x)=\max_{\boldsymbol y\in \mathcal P_\kappa}\ip{\boldsymbol p}{\boldsymbol y-\boldsymbol x}$. Then
\[
H_h(\boldsymbol x)\ge\alpha P_h-\xi M-\gap(\boldsymbol p,\boldsymbol x)
\]
for every legal future. In particular, for $\eta\ge0$, $\gap(\boldsymbol p,\boldsymbol x)\le\eta M$ gives $H_h(\boldsymbol x)\ge(\alpha-\xi-\eta)P_h$. The same conclusions hold if the price inequality is required only for a specified subfamily of legal futures.
\end{lemma}
\begin{proof}
Choose a maximizer $O$ of $P_h$. Since $\one_O\in \mathcal P_\kappa$, the price term is at least $-\gap(\boldsymbol p,\boldsymbol x)$. Also $P_h\ge M$, because $h\ge f$ and $\kappa\ge1$.
\end{proof}

\subsection{Future-Uniform Prices from a Current Potential}\label{sec:scale}

The linear comparison term in the price inequality can be combined with an exact modular identity, as used in \cref{sec:curvature} to preserve a modular component with coefficient one. Coverage instead has the stronger price inequality in \cref{lem:coverage-prices}.

Non-oblivious potentials were developed for maximum coverage~\citep{FW12} and general monotone submodular maximization~\citep{FW14}. In the continuous setting, \citep[Lemma~2 and Theorems~1--2]{ZDC+22} integrates weighted gradients over scales in $[0,1]$ to obtain approximation guarantees from stationarity. Here gradients of the current Poisson response certify comparisons after adjoining any fixed future set, using a potential whose definition and optimization are independent of that set.

Let $g$ be a normalized monotone submodular function on $V$. For a current set $X\subseteq V$, a fixed set $R\subseteq V$, and $\boldsymbol z\in[0,1]^X$, let $F_R^g(\boldsymbol z)=\E[g(S_{\boldsymbol z}\cup R)]$ be the Bernoulli multilinear extension~\citep{CCPV11}, where $S_{\boldsymbol z}$ includes each $i\in X$ independently with probability $z_i$, and define
$H_R^g(\boldsymbol x)=F_R^g(\one-e^{-\boldsymbol x})$.
All exponentials are coordinatewise. The Poisson extension is defined for every nonnegative vector of sampling parameters, including those outside $[0,1]^n$. We integrate the expected current value over different sampling scales to obtain a potential whose normalized gradient gives the required prices. For $1\le \Lambda\le\sqrt2$, define the \emph{scale potential}
\[
\Phi_\Lambda^g(\boldsymbol x)=\int_0^\Lambda\frac{H_\varnothing^g(t\boldsymbol x)}t\,dt.
\]
The integrand has a continuous limit at zero.

\begin{lemma}[Prices from the potential]\label{lem:scale}
Let $1\le\Lambda\le\sqrt2$. For every nonnegative $\boldsymbol x$, every $O\subseteq X$, and every fixed $R$, including $R\cap X\ne\varnothing$,
\begin{equation}\label{eq:scale}
(1+\Lambda)H_R^g(\boldsymbol x)-\Lambda{}g(O\cup R)
\ge\ip{\nabla\Phi_\Lambda^g(\boldsymbol x)}{\boldsymbol x-\one_O}.
\end{equation}
In particular, for the current function $f$ and $\boldsymbol x\in\mathcal P_\kappa$, the vector $\boldsymbol p_\Lambda(\boldsymbol x)=\nabla\Phi_\Lambda^f(\boldsymbol x)/(1+\Lambda)$ is a universal $\Lambda/(1+\Lambda)$-price determined only by current values. Taking $\Lambda=\sqrt2$ gives a universal $\bconst$-price.
\end{lemma}
\begin{proof}
\emph{Current marginals control the future gap.}
Define $q_{\varnothing}(t)=H_\varnothing^g(t\boldsymbol x)$, $q(t)=H_R^g(t\boldsymbol x)$, and $P=g(O\cup R)$. For a Poisson union $Z$ with sampling parameters $t\boldsymbol x$, monotonicity and diminishing returns give
\[
\sum_{i\in O\setminus Z}g(i\mid Z)
\ge g(O\mid Z)\ge g(O\mid Z\cup R)\ge P-g(Z\cup R).
\]
The expected left side involves only current marginals. Indeed,
\[
\partial_iH_\varnothing^g(t\boldsymbol x)
=e^{-tx_i}\E g(i\mid Z_{-i,t\boldsymbol x}),
\]
where $Z_{-i,t\boldsymbol x}$ samples all coordinates except $i$. Since $e^{-tx_i}=\Prob(i\notin Z)$ and the coordinates are independent,
\[
\sum_{i\in O}\partial_iH_\varnothing^g(t\boldsymbol x)
=\E\!\left[\sum_{i\in O\setminus Z}g(i\mid Z)\right]
\ge P-q(t).
\]

\emph{Integrating along the sampling ray.}
The factor $1/t$ in the potential is chosen so that the chain-rule factor $t$ cancels:
\[
\nabla\Phi_\Lambda^g(\boldsymbol x)=\int_0^\Lambda\nabla H_\varnothing^g(t\boldsymbol x)\,dt,
\qquad
\ip{\nabla\Phi_\Lambda^g(\boldsymbol x)}{\boldsymbol x}=\int_0^\Lambda q_{\varnothing}'(t)\,dt=q_{\varnothing}(\Lambda),
\]
since $q_{\varnothing}(0)=g(\varnothing)=0$. Integrating the marginal bound gives
\[
\ip{\nabla\Phi_\Lambda^g(\boldsymbol x)}{\boldsymbol x-\one_O}
\le q_{\varnothing}(\Lambda)+\int_0^\Lambda{}q(t)\,dt-\Lambda{}P
\le q(\Lambda)+\int_0^\Lambda{}q(t)\,dt-\Lambda{}P,
\]
where the last inequality uses monotonicity.

Adjoining a fixed $R$ preserves monotonicity and submodularity, even when $R\cap X\ne\varnothing$. Every entry of the Hessian of its Poisson extension is nonpositive: diagonals are minus first derivatives, and mixed entries are weighted discrete second differences. Since $\boldsymbol x\ge0$, monotonicity and these signs imply that $q$ is nondecreasing and concave as a \emph{scalar} function of $t$. Its tangent at one gives $q(t)\le q(1)+(t-1)q'(1)$ for $t\ge0$. Applying this bound at $\Lambda$ and under the integral yields
\[
\begin{aligned}
q(\Lambda)+\int_0^\Lambda q(t)\,dt
&\le(1+\Lambda)q(1)
+\left(\Lambda-1+\int_0^\Lambda(t-1)\,dt\right)q'(1)\\
&=(1+\Lambda)q(1)+(\Lambda^2/2-1)q'(1)\\
&\le(1+\Lambda)q(1),
\end{aligned}
\]
where the last inequality uses $\Lambda\le\sqrt2$ and $q'(1)\ge0$. Together with the integrated bound, this proves \eqref{eq:scale}.

Encode each fixed legal future $h$ by $r\notin X$, as discussed in \citep[Section~5.1, before Lemma~5.1]{DFL+25}, via $\widetilde f(S)=f(S)$ and $\widetilde f(S\cup\{r\})=h(S)$ for $S\subseteq X$. The marginal of $r$ is $h(S)-f(S)\ge0$, and adding $i\in X\setminus S$ changes it by $h(i\mid S)-f(i\mid S)\le0$. By \cref{eq:legal}, this extension is monotone submodular, so \eqref{eq:scale} applies with $R=\{r\}$. Although the extension depends on $h$, its restriction to current subsets is always $f$, so $\Phi_\Lambda^{\widetilde f}=\Phi_\Lambda^f$. Hence the same vector $\boldsymbol p_\Lambda(\boldsymbol x)=\nabla\Phi_\Lambda^f(\boldsymbol x)/(1+\Lambda)$ satisfies the price inequality for every legal future $h$ and every comparison set $O\subseteq X$. Singleton bounds give $p_{\Lambda,i}\in[0,M_i]$. For this certificate, $\Lambda=1$ gives coefficient $1/2$. Taking $\Lambda=\sqrt2$, the largest scale allowed by this tangent estimate, gives coefficient $\sqrt2/(1+\sqrt2)=\bconst$.
\end{proof}

The following corollary estimates prices at a prescribed point. Finding a point with a small price gap is the separate optimization task described below.

\begin{corollary}[Prices at a prescribed point]\label{cor:computable-prices}
Fix $0\le\alpha<\bconst$. Given $\boldsymbol x\in[0,1]^n$, $\xi>0$, and $\delta_{\rm fail}\in(0,1)$, a routine using only current value queries returns a vector $\widehat{\boldsymbol p}\in\prod_i[0,M_i]$ on every outcome. With probability at least $1-\delta_{\rm fail}$, the same vector satisfies
\[
H_h(\boldsymbol x)-\alpha h(O)\ge\ip{\widehat{\boldsymbol p}}{\boldsymbol x-\one_O}-\xi M
\quad\text{for all legal }h\text{ and all }O\subseteq X.
\]
The routine uses polynomially many current value queries and arithmetic operations in $n$, $1/\xi$, and $\log(1/\delta_{\rm fail})$.
\end{corollary}
\begin{proof}
Set $\Lambda=\sqrt2$ and estimate $\boldsymbol p_\Lambda(\boldsymbol x)$ from bounded current marginals. Clipping coordinate $i$ to $[0,M_i]$ cannot increase its error. Coordinate accuracy $\xi M/(2n)$ controls every price term simultaneously, because $\|\boldsymbol x-\one_O\|_1\le n$. The exact price has coefficient $\bconst\ge\alpha$. Appendix~\ref{app:stationarity} gives the prescribed-point estimator and its sample bound.
\end{proof}

For $J=\Phi_{\sqrt2}^f/(1+\sqrt2)$ and $\eta,\delta_{\rm fail}\in(0,1)$, the current-value routine in \cref{lem:first-order} uses bounded marginal estimates and the range and smoothness bounds to find a feasible $\boldsymbol x$ with $\gap(\nabla J(\boldsymbol x),\boldsymbol x)\le\eta M$, except with probability $\delta_{\rm fail}$. The routine uses polynomially many current value queries and requires only first-order stationarity~\citep{LJ16}. The same successful estimation event certifies every legal future through \cref{lem:price-gap}.

For a legal future $h$, put $F_h(\boldsymbol z)=\E h(S_{\boldsymbol z})$, so $F_h(\boldsymbol x)\ge H_h(\boldsymbol x)$ on $[0,1]^n$. Randomized pipage rounding returns a set $A$ with $|A|\le\kappa$ on every outcome and $\E_{\rm round}h(A)\ge F_h(\boldsymbol x)\ge H_h(\boldsymbol x)$ for each fixed legal future~\citep{AS04,CVZ10}. Its mean-preserving pair exchanges depend only on $\boldsymbol x$, and convexity along these exchanges gives the expected-value guarantee (\cref{lem:pipage}). Nonnegativity on optimization failure then gives $\E h(A)\ge(\bconst-\eta-\delta_{\rm fail})P_h$, using the same current-only distribution for every fixed legal future.

For coverage, we strengthen the price inequality and use independent categorical slots with concavity to protect intermediate replacement states in \cref{sec:structure}. For bounded curvature, we apply the scale inequality to the residual submodular part and keep the modular term linear before rounding in \cref{sec:curvature}.

\section{Coverage Without Its Representation}\label{sec:structure}

Coverage improves both solution quality and replacement beyond the general price inequality of \cref{sec:scale}. The gradient of its current Poisson extension certifies coefficient $\aconst=1-e^{-1}$, improving endpoint quality. For each fixed set, concavity of the Poisson extension in the fractional point then protects intermediate states as independent sampled entries are replaced. The sampled sets stay feasible throughout a fixed replacement schedule. We prove the stronger price inequality before comparing sampling schemes and analyzing replacement.

\begin{theorem}[Coverage without its representation]\label{thm:coverage}
Suppose $f$ is a nonnegative weighted coverage function on $V$. For every $\varepsilon\in(0,1-1/e)$, there is a randomized algorithm that uses polynomially many current value queries per arrival, without access to the underlying features or weights, and maintains $S_t\subseteq X_t$ with $|S_t|\le k$ such that, at every fixed time $t$,
\[
\E f(S_t)\ge\left(1-\frac1e-\varepsilon\right)\OPT_k(X_t).
\]
The worst-case symmetric recourse is at most $4\lceil2/\varepsilon\rceil+2$ per arrival, for every random outcome.
\end{theorem}

\begin{remarkstar}[Computational optimality]
By the hardness of maximum coverage \citep{Fei98}, any fixed improvement above $\aconst$ in randomized polynomial time on explicit instances would imply $\mathsf{NP}\subseteq\mathsf{BPP}$, even without recourse restrictions. \Cref{prop:coverage-hardness} gives the reduction.
\end{remarkstar}

\subsection{The Current Gradient Certifies Coverage Futures}

For the same coverage objective $f$ on the entire stream, use $h_R(S)=f(S\cup R)$ for fixed $R\subseteq V$. We prove the price inequality one feature at a time, using the weights $w_a$ and covering sets $C_a$ from \cref{sec:model-cores}. Let $H=H_\varnothing$ be the current Poisson extension.

\begin{lemma}[Coverage prices]\label{lem:coverage-prices}
For every $\boldsymbol{x}\ge0$, every $O\subseteq X$, and every fixed set $R\subseteq V$, possibly intersecting $X$,
\begin{equation}\label{eq:coverage-certificate}
H_R(\boldsymbol{x})-\aconst f(O\cup R)\ge\ip{\nabla H(\boldsymbol{x})}{\boldsymbol{x}-\one_O}.
\end{equation}
\end{lemma}
\begin{proof}
For one feature of positive weight, divide by its weight and define $z=\sum_{i\in X\cap C_a}x_i$, $o_a=|O\cap C_a|$, and $b_R=\one\{R\cap C_a\ne\varnothing\}$. Its price contribution is $(z-o_a)e^{-z}$. If $b_R=o_a=0$, the claim is $1-e^{-z}\ge ze^{-z}$. If $b_R=0$ and $o_a\ge1$, it follows from $(z-o_a+1)e^{-z}\le ze^{-z}\le e^{-1}$. If $b_R=1$, it follows from $(z-o_a)e^{-z}\le ze^{-z}\le e^{-1}$. Summing over the features gives the price inequality, with coefficient $\aconst$ from the scalar bound $\max_{z\ge0}ze^{-z}=e^{-1}$. Both $H$ and $\nabla H$ are determined by the function values on current sets.
\end{proof}

The gradient can be estimated using current value queries through
\begin{equation}\label{eq:coverage-gradient}
\partial_iH(\boldsymbol{x})=e^{-x_i}\E[f(Z_{-i}\cup\{i\})-f(Z_{-i})],
\end{equation}
where each $j\ne i$ enters $Z_{-i}$ independently with probability $1-e^{-x_j}$. Each sample uses two current value queries, and its marginal is at most $M=\max_i f(\{i\})$. The first-order routine in \cref{lem:first-order} finds $\boldsymbol{x}\in \mathcal P_\kappa$ with first-order gap at most $\eta M$, except with probability $\delta_{\rm fail}$. The first-order success event involves only current marginal estimates, so the price inequality gives $H_R(\boldsymbol{x})\ge\aconst P_R-\eta M$ on that same event simultaneously for every fixed $R\subseteq V$. Since $M\le P_R$ for $\kappa\ge1$ and $H_R(\boldsymbol{x})\ge0$ even on failure, averaging gives
\begin{equation}\label{eq:coverage-core}
\E H_R(\boldsymbol{x})\ge(\aconst-\eta-\delta_{\rm fail})P_R,
\qquad P_R=\max_{O\subseteq X,\ |O|\le\kappa}f(O\cup R).
\end{equation}
The algorithm keeps the fractional point for sampling, allowing replacement to interpolate between such points while preserving their guarantees.

\subsection{Concavity and Independent Slots}

For fixed $R$, the coverage formula in \citep[Section~4.2, Eq.~(8)]{DRY11} becomes
\begin{equation}\label{eq:coverage-concavity}
H_R(\boldsymbol{x})=\sum_{a:R\cap C_a\ne\varnothing}w_a+
\sum_{a:R\cap C_a=\varnothing}w_a\left(1-e^{-\sum_{i\in X\cap C_a}x_i}\right).
\end{equation}
The features covered by $R$ contribute constants, so $H_R$ remains concave in $\boldsymbol{x}$. Unlike the scalar ray concavity used in \cref{lem:scale}, this vector concavity protects every interpolation between two fractional points. Independent categorical slots, each holding an element or a null, implement the interpolation throughout an entire block, including when the recent set overlaps the newer current set as allowed by \cref{lem:overlap}.

Represent $\boldsymbol{x}\in \mathcal P_\kappa$ by $\kappa$ independent slots, each choosing $i$ with probability $x_i/\kappa$ and a null otherwise, so their distinct nonnull entries form a feasible set. This sampling principle and its domination of the Poisson extension appear in \citep[Section~5, Lemma~5.6]{BNW25}. The heterogeneous-slot form below follows from the argument in their Lemmas~5.4--5.6. We give a direct replacement proof for arbitrary fixed $R$.

\begin{lemma}[Categorical domination]\label{lem:slots}
Fix a coordinate set $X$ and a deterministic set $R$, with arbitrary overlap. Let $I_1,\ldots,I_\kappa$ be independent categorical draws in $X$ or a null, with $\Prob(I_s=i)=p_i^{(s)}$ and the remaining probability assigned to a null. Define $z_i=\sum_s p_i^{(s)}$. Their distinct nonnull values $A$ satisfy $|A|\le\kappa$ on every outcome and
\[
\E f(A\cup R)\ge H_R(\boldsymbol{z})
\]
for every monotone submodular function on $X\cup R$.
\end{lemma}
\begin{proof}
We compare categorical and Poisson sampling one slot at a time, fixing $s$ and writing $p_i=p_i^{(s)}$. Condition on the union $U$ of $R$ and all other draws, so the slot gives conditional expectation $f(U)+\sum_i p_i f(i\mid U)$. Replace it by independent counts $N_i\sim\Poi(p_i)$, including $i$ whenever $N_i>0$. By submodularity, the replacement has expected value at most
\[
f(U)+\sum_i(1-e^{-p_i})f(i\mid U)
\le f(U)+\sum_i p_i f(i\mid U).
\]
Continue replacing slots with independent counts for each slot and coordinate. The aggregated counts are independent Poisson variables of means $z_i$, so their union has response $H_R(\boldsymbol{z})$, with zero marginal throughout for coordinates in $R$.
\end{proof}

For two fractional points $\boldsymbol{x},\boldsymbol{y}$, take $\kappa-j$ old slots and $j$ new slots, each with the corresponding coordinate probabilities divided by $\kappa$. If these draws are independent conditional on $\boldsymbol{x},\boldsymbol{y}$, their distinct nonnull values $A_j$ satisfy, with $\lambda=j/\kappa$,
\begin{equation}\label{eq:slot-interpolation}
\begin{aligned}
\E_{\rm slot}[f(A_j\cup R)\mid \boldsymbol{x},\boldsymbol{y}]
&\ge H_R((1-\lambda)\boldsymbol{x}+\lambda \boldsymbol{y})\\
&\ge(1-\lambda)H_R(\boldsymbol{x})+\lambda H_R(\boldsymbol{y}).
\end{aligned}
\end{equation}
Categorical domination gives the first inequality and concavity the second, so each intermediate state preserves the endpoint interpolation in conditional expectation.

\subsection{Migration with \texorpdfstring{$O(\varepsilon^{-1})$}{O(1/epsilon)} Recourse}\label{sec:slot-migration}

We use the current sets and recent sets in \citep{DFL+25}, with the block convention of \cref{lem:checkpoint}. Inequality~\eqref{eq:slot-interpolation} protects each intermediate state in conditional expectation, so replacement no longer needs a short random window. Spreading replacement across the block reduces per-arrival changes from $O(B^2)$ to $O(B)$. The replacement times and position order are deterministic, while the algorithm remains randomized. The following conversion applies whenever the Poisson extension is concave after adjoining any fixed set.

\begin{lemma}[Slot conversion]\label{lem:slot-conversion}
Consider a class of monotone submodular functions on $V$. Suppose that, for every function in this class, the Poisson extension $H_R$ is concave for each fixed $R\subseteq V$. Let $0<\alpha\le1$, and suppose that, for each current set $Y$ and capacity $\kappa\ge1$, a routine using only current value queries returns a random vector $\boldsymbol{x}\in \mathcal P_\kappa(Y)$ satisfying
\begin{equation}\label{eq:slot-core-interface}
\E H_R(\boldsymbol{x})\ge\alpha\max_{O\subseteq Y,\ |O|\le\kappa} f(O\cup R)
\end{equation}
for every fixed $R\subseteq V$, with arbitrary overlap with $Y$ and with the same law of $\boldsymbol{x}$ for all $R$. Computations of fractional points use randomness independent of all slot draws and receive no feedback from the realized slots or maintained set. For an integer $B\ge4$ and $k\ge2B$, these fractional points give an online algorithm with expected coefficient $\alpha(1-2/B)$ at each fixed time and worst-case symmetric recourse at most $4B+2$.
\end{lemma}
\begin{proof}
\emph{The update rule.} Let $L=\lfloor k/B\rfloor$,
$\kappa=k-2L>0$, and $c_{\rm step}=\lceil\kappa/L\rceil$.
Since $k\ge2B$ gives $L\ge k/(2B)$, we have $c_{\rm step}\le2B$.
Initialize $\boldsymbol{x}^{(0)}=\boldsymbol{0}$, the all-null slot
tuple $A_0=(\bot,\ldots,\bot)$, and
$\mathcal T_0=S_0=\varnothing$. At time $t$, update $S_t$ by
step~1 or~2, then perform step~3 whenever its condition holds.
\begin{enumerate}[label=\arabic*.]
\item If $t\le L$, set $\mathcal T_t=S_t=X_t$.
\item If $t>L$, express $t=qL+\iota$ with $q\ge1$ and
$1\le\iota\le L$, and set $j_t=\min\{\kappa,c_{\rm step}\iota\}$.
Form
\[
\begin{aligned}
\mathcal A_t&=\bigl(\{A_q[i]:1\le i\le j_t\}
\cup\{A_{q-1}[i]:j_t<i\le\kappa\}\bigr)\setminus\{\bot\},\\
R_t&=X_t\setminus X_{(q-1)L}.
\end{aligned}
\]
Set $\mathcal T_t=\mathcal A_t\cup R_t$. To update $S_t$, insert
$\mathcal T_t\setminus S_{t-1}$ and, only if capacity is exceeded,
delete enough elements outside $\mathcal T_t$ in a fixed order to
restore capacity $k$. This is the rule of \cref{lem:lazy}.
\item If $t$ is a multiple of $L$, including $t=L$, compute a
fractional point $\boldsymbol{x}^{(t/L)}$ on $X_t$ for the next block.
With fresh independent sampling randomness, draw the $\kappa$ entries
of $A_{t/L}$, each choosing $i\in X_t$ with probability
$x_i^{(t/L)}/\kappa$ and a null otherwise.
\end{enumerate}

\emph{Feasibility and recourse.}
For $t>L$, $|\mathcal A_t|\le\kappa$ and $|R_t|\le2L$, so
$|\mathcal T_t|\le k$. The update maintains
$\mathcal T_t\subseteq S_t\subseteq X_t$ and $|S_t|\le k$.
Within a block, at most $c_{\rm step}$ positions change and one new
element arrives, so $|\mathcal T_t\setminus\mathcal T_{t-1}|
\le c_{\rm step}+1$. Since $c_{\rm step}L\ge\kappa$, replacement
is complete at the end of each block. This complete tuple becomes the
old tuple in the next block, whose first update replaces at most
$c_{\rm step}$ positions. Resetting the recent set only removes old
elements and adds the new arrival, so the insertion bound also holds
at block boundaries. Such target removals need not be mirrored in
$S_t$. At the first transition,
$R_{L+1}=X_{L+1}$, so $\mathcal T_{L+1}=X_{L+1}$ adds only the new
arrival, as does each initialization step. Because
$\mathcal T_{t-1}\subseteq S_{t-1}$, the actual number of insertions
is $|\mathcal T_t\setminus S_{t-1}|
\le|\mathcal T_t\setminus\mathcal T_{t-1}|
\le c_{\rm step}+1$. By \cref{lem:lazy}, deletions are needed only
to restore capacity and never outnumber insertions, giving
\begin{equation}\label{eq:linear-recourse}
|S_t\triangle S_{t-1}|\le2(c_{\rm step}+1)\le4B+2
\end{equation}
on every outcome.

\emph{Value at every intermediate state.}
For $t\le2L$, we have $\mathcal T_t=S_t=X_t$: this follows from
initialization when $t\le L$, and from $R_t=X_t\setminus X_0=X_t$
in step~2 when $L<t\le2L$.
Since $2L\le k$, these states are optimal.
Now fix $t=qL+\iota$ with $q\ge2$ and $1\le\iota\le L$.
The oblivious input stream fixes $R_t$ independently of all algorithmic
randomness. Extend $\boldsymbol{x}^{(q-1)}$ by zero to $X_{qL}$ and
let $\mathcal G=\sigma(\boldsymbol{x}^{(q-1)},\boldsymbol{x}^{(q)})$.
The no-feedback assumption and fresh sampling randomness independent
of the fractional-point computations ensure that, conditional on
$\mathcal G$, the $\kappa-j_t$ old slots and $j_t$ new slots selected
by the fixed replacement schedule and position order are independent. With
$\lambda=j_t/\kappa$, \eqref{eq:slot-interpolation} gives
\begin{equation}\label{eq:migration-conditional-value}
\E_{\rm slot}[f(\mathcal T_t)\mid\mathcal G]
\ge(1-\lambda)H_{R_t}(\boldsymbol{x}^{(q-1)})
+\lambda H_{R_t}(\boldsymbol{x}^{(q)}).
\end{equation}
Both endpoint responses use this same $R_t$.
If $O_t$ attains $\OPT_\kappa(X_t)$, then
$O_t\setminus R_t\subseteq X_{(q-1)L}\subseteq X_{qL}$ and
$|O_t\setminus R_t|\le\kappa$. Thus, for either
$Y\in\{X_{(q-1)L},X_{qL}\}$,
\[
\max_{O\subseteq Y,\ |O|\le\kappa}f(O\cup R_t)
\ge f((O_t\setminus R_t)\cup R_t)
\ge f(O_t)=\OPT_\kappa(X_t).
\]
Although some elements of $R_t$ arrive after the fractional points
are computed, \eqref{eq:slot-core-interface} applies to this fixed
set with the same law of each point. Its allowance for overlap also
covers $R_t\cap X_{qL}$. Since $R_t$ and $\lambda$ are deterministic
for this fixed $t$, we can average \eqref{eq:migration-conditional-value}
and apply \eqref{eq:slot-core-interface} to each point's marginal
distribution, even when the points are correlated. Together with
monotonicity, this yields
\[
\E f(S_t)\ge\E f(\mathcal T_t)\ge\alpha\OPT_\kappa(X_t).
\]
Finally, a uniform $\kappa$-subset of an optimal $k$-set, padded by
nulls if necessary, has expected value at least
$(\kappa/k)\OPT_k(X_t)$ by \eqref{eq:equal-marginal-sampling}.
Therefore
\[
\E f(S_t)\ge\alpha\OPT_\kappa(X_t)
\ge\alpha\frac\kappa k\OPT_k(X_t)
\ge\alpha(1-2/B)\OPT_k(X_t).
\]
The same fixed-time argument applies to a final incomplete block.
\end{proof}

\begin{proof}[Proof of the upper bound in \cref{thm:coverage}]
We use direct recomputation for small capacity and slot replacement otherwise. Take $B=\lceil2/\varepsilon\rceil\ge4$. If $k<2B$, ordinary greedy after each arrival achieves coefficient at least $\aconst$ with at most $2k<4B$ symmetric changes. Otherwise use the schedule above with the fractional points for coverage from \cref{eq:coverage-core} and $\eta=\delta_{\rm fail}=\varepsilon/12$, giving $\alpha=\aconst-\varepsilon/6$.

With exact categorical draws, \cref{lem:slot-conversion} gives coefficient at least
\[
(\aconst-\varepsilon/6)(1-2/B)
\ge\aconst-(\aconst+\nicefrac16)\varepsilon
>\aconst-\varepsilon.
\]
Here $2/B\le\varepsilon$ and $\aconst+\nicefrac16<1$. On every outcome, the algorithm obeys \cref{eq:linear-recourse}, giving the stated worst-case recourse bound. The query bound follows from \cref{eq:oracle-work}. Appendix~\ref{app:coverage-details} proves the separate computational lower bound.
\end{proof}

\Needspace{10\baselineskip}
\section{Modular Preservation and Bounded Curvature}\label{sec:curvature}

The general price inequality in \cref{sec:scale} yields an exact dependence on curvature when applied only to the residual submodular part, while a known modular component keeps coefficient one. We install feasible sets produced by pair rounding through checkpoints, analyzing the objective used for comparison below. An upper bound~\citep{CC84} $\vartheta\in[0,1]$ on the total curvature of $f$ means
\begin{equation}\label{eq:curvature-def}
f(i\mid S)\ge(1-\vartheta)f(\{i\})
\qquad(S\subseteq V,\ i\in V\setminus S).
\end{equation}

\begin{theorem}[Sharp threshold under bounded curvature]\label{thm:curvature-main}
Suppose $f$ has total curvature at most a known $\vartheta\in[0,1]$ over the entire stream, and let $\rho_\vartheta=1-(\sqrt2-1)\vartheta$. For every $\varepsilon>0$, there is a randomized algorithm using polynomially many current value queries per arrival that maintains $S_t\subseteq X_t$ with $|S_t|\le k$ and satisfies, at every fixed time $t$,
\[
\E f(S_t)\ge(\rho_\vartheta-\varepsilon)\OPT_k(X_t).
\]
The algorithm has $O(\varepsilon^{-2})$ worst-case symmetric recourse per arrival, for every random outcome. The threshold is sharp. For every fixed $\vartheta\in(0,1]$ and $\zeta>0$, there exist constants $c_{\vartheta,\zeta},d_{\vartheta,\zeta}>0$ such that the following holds for all sufficiently large $k$. Consider any randomized algorithm in this model that, on every instance of at most $d_{\vartheta,\zeta}k$ elements with curvature at most $\vartheta$ and every random outcome, makes at most $e^{c_{\vartheta,\zeta}k}$ queries before the final arrival and at most $c_{\vartheta,\zeta}k$ symmetric changes there. On some instance in this class fixed independently of its randomness, its expected approximation ratio at the final arrival is less than $\rho_\vartheta+\zeta$. This holds even with unrestricted queries and computation after the final arrival.
\end{theorem}

At $\vartheta=1$, the threshold reduces to $2-\sqrt2$, recovering the sharp general-case threshold of \cref{thm:main}. At $\vartheta=0$, the objective is modular, and an exact optimum can be maintained with at most two symmetric changes per arrival. The upper bound preserves the modular component of the objective, while the matching lower bound augments the general hard instance with a modular term. Throughout, $\OPT_k(u,Z)=\max_{U\subseteq Z,\ |U|\le k}u(U)$.

\subsection{Preserving a Modular Component}

For $f=g+\ell$, we seek $\ell(O)+\bconst g(O)$ for the same feasible comparison set $O$. Applied to $f$, the guarantee of \cref{sec:cores} certifies $\bconst f(O)$, losing $(1-\bconst)\ell(O)$ relative to this mixed benchmark. Offline, the modular-preservation principle in \citep[Theorem~3.1]{SVW17} achieves $\ell(O)+(1-1/e)g(O)$ up to an arbitrarily small additive error, yielding the curvature threshold $1-\vartheta/e$. Feldman~\citep{Fel21} gives a guess-free algorithm for the offline benchmark. Here we preserve coefficient one on $\ell$ in a static core computed using only current queries, with coefficient $\bconst$ on $g$ after each fixed future set. We then transfer this mixed benchmark through checkpoints in \cref{lem:hybrid-checkpoint}.

\begin{theorem}[Preserving a modular component]\label{thm:hybrid}
Suppose $f=g+\ell$ on $2^V$, where $g$ is normalized monotone submodular and $\ell$ is nonnegative modular. The weight $\ell_i$ is available when $i$ arrives, and $g$ can be accessed through current value queries. For every $\varepsilon>0$, there is a randomized algorithm that uses polynomially many current value queries per arrival and $O(\varepsilon^{-2})$ worst-case symmetric recourse and satisfies, at each fixed time,
\begin{equation}\label{eq:hybrid-goal}
\E f(S_t)\ge
\max_{O\subseteq X_t,\ |O|\le k}\{\ell(O)+\bconst g(O)\}
-\varepsilon\OPT_k(f,X_t).
\end{equation}
\end{theorem}

\paragraph{A static mixed core.}
Fix a current set $Y$ and capacity $\kappa\ge1$, and define the \emph{mixed objective} $f_{\rm cmp}=\ell+\bconst g$, so $f_{\rm cmp}$ is monotone submodular and $0\le f_{\rm cmp}\le f$. Set $\Lambda=\sqrt2$ and use
\[
\Xi(\boldsymbol{x})=\boldsymbol{\ell}\cdot\boldsymbol{x}+\frac{\Phi_\Lambda^g(\boldsymbol{x})}{1+\Lambda},
\qquad
L_R(\boldsymbol{x})=\ell(R)+\boldsymbol{\ell}\cdot\boldsymbol{x}+H_R^g(\boldsymbol{x}),
\qquad \boldsymbol{x}\in \mathcal P_\kappa(Y).
\]
The algorithm optimizes only $\Xi$, using current values. The response $L_R$ appears only in the analysis. We keep the modular term linear because Poisson sampling of a unit-weight modular coordinate at $x_i=1$ reduces its contribution to $1-1/e$. For every $O\subseteq Y$ and fixed $R\subseteq V\setminus Y$, disjointness gives the exact modular identity
\[
\ell(R)+\boldsymbol{\ell}\cdot\boldsymbol{x}-\ell(O\cup R)
=\ip{\boldsymbol{\ell}}{\boldsymbol{x}-\one_O}.
\]
Adding this identity to the inequality of \cref{lem:scale} for $g$, divided by $1+\Lambda$, gives
\begin{equation}\label{eq:hybrid-certificate}
L_R(\boldsymbol{x})-f_{\rm cmp}(O\cup R)
\ge\ip{\nabla\Xi(\boldsymbol{x})}{\boldsymbol{x}-\one_O}.
\end{equation}
Thus the two prices use the same comparison set $O$, with coefficient one on its modular value. For the Bernoulli multilinear extension, disjointness gives $F_R^f(\boldsymbol{x})=\ell(R)+\boldsymbol{\ell}\cdot\boldsymbol{x}+F_R^g(\boldsymbol{x})$. Monotonicity and $1-e^{-x_i}\le x_i$ give $F_R^g(\boldsymbol{x})\ge H_R^g(\boldsymbol{x})$, and hence $F_R^f(\boldsymbol{x})\ge L_R(\boldsymbol{x})$.

An empty current set uses the empty sampled set. Otherwise let $M=\max_{i\in Y}f(\{i\})$ and $P_R=\max_{O\subseteq Y,\ |O|\le\kappa}f(O\cup R)$. If $M=0$, current marginals vanish and the empty sampled set suffices. For the remaining case, the potential has range at most $\kappa M$, gradient coordinates in $[0,f(\{i\})]$, and Hessian entries of magnitude at most $M$. By \cref{lem:first-order}, we can find, using only current value queries, a point with gap at most $\eta M$, except with probability $\delta_{\rm fail}$, using effective capacity $\min\{\kappa,|Y|\}$. On the same successful first-order event, for every feasible $O$ and fixed disjoint future set,
\begin{equation}\label{eq:hybrid-pointwise-core}
F_R^f(\boldsymbol{x})\ge L_R(\boldsymbol{x})
\ge f_{\rm cmp}(O\cup R)-\eta M.
\end{equation}
Mean-preserving pair rounding, whose law depends only on $\boldsymbol{x}$, keeps the expected modular contribution equal to $\boldsymbol{\ell}\cdot\boldsymbol{x}$ and does not decrease $F_R^f$ in expectation (\cref{lem:pipage}). Since $M\le P_R$ and $f_{\rm cmp}\le f$, optimization failure costs at most $\delta_{\rm fail}P_R$. The sampled set therefore satisfies
\begin{equation}\label{eq:hybrid-static-core}
\E f(A_Y\cup R)\ge
\max_{O\subseteq Y,\ |O|\le\kappa}f_{\rm cmp}(O\cup R)-\delta_{\rm core}P_R,
\qquad
\delta_{\rm core}=\eta+\delta_{\rm fail}.
\end{equation}

\paragraph{Transferring the mixed benchmark.}
The checkpoint schedule comes from \citep[Theorem~3.2 and Section~4]{DFL+25}, with the symmetric-recourse implementation of \cref{lem:checkpoint}. Its multiplicative interface evaluates the sampled set and the benchmark with the same function. In \cref{eq:hybrid-static-core}, the sampled set is evaluated by $f$, the benchmark by $f_{\rm cmp}$, and the additive error by $f$. We retain the schedule and adapt its value analysis: capacity reduction applies to $f_{\rm cmp}$ with a fixed recent suffix, and averaging over replacement windows counts the additive error once.

\begin{lemma}[Hybrid checkpoint conversion]\label{lem:hybrid-checkpoint}
Let $f,f_{\rm cmp}$ be nonnegative monotone submodular functions on $V$, with $f_{\rm cmp}\le f$. Fix an integer $B\ge4$, $\delta_{\rm core}\ge0$, and $k\ge4B^2$, and define $L=\lfloor k/B\rfloor$, $\kappa=k-2L$, and $W=\lfloor L/B\rfloor$. Suppose a sampler using only current value queries on each current set $Y$ returns a random set $A_Y\subseteq Y$ with $|A_Y|\le\kappa$ on every outcome such that, for every fixed $R\cap Y=\varnothing$ with $\kappa+|R|\le k$,
\begin{equation}\label{eq:hybrid-static-interface}
\E f(A_Y\cup R)\ge
\max_{O\subseteq Y,\ |O|\le\kappa}f_{\rm cmp}(O\cup R)-\delta_{\rm core} P_f(Y,R),
\qquad
P_f(Y,R)=\max_{O\subseteq Y,\ |O|\le\kappa}f(O\cup R).
\end{equation}
Choose each replacement window uniformly among the $B$ windows, independently of the sampler. The checkpoint schedule has worst-case symmetric recourse at most $8B^2+2$ and, at each fixed time,
\begin{equation}\label{eq:hybrid-checkpoint-conclusion}
\E f(S_t)\ge\OPT_k(f_{\rm cmp},X_t)-\left(\delta_{\rm core}+\frac3B\right)\OPT_k(f,X_t).
\end{equation}
\end{lemma}
\begin{proof}
The capacity and recourse invariants of \cref{lem:checkpoint} still hold. It remains to transfer the mixed benchmark and its additive error.

\emph{Disjoint suffixes.} Fix $t=qL+\iota$ with $q\ge1$ and $1\le \iota\le L$. Before and after replacement, respectively, take
\[
(Y^-,R^-)=(X_{(q-1)L},X_t\setminus X_{(q-1)L}),
\qquad
(Y^+,R^+)=(X_{qL},X_t\setminus X_{qL}).
\]
Each pair partitions $X_t$ into a prefix and a disjoint suffix, and $\kappa+|R^\pm|\le k$. The maintained set contains $A_{Y^-}\cup R^-$ before replacement and $A_{Y^+}\cup R^+$ afterwards. In the latter case the actual recent set may be larger, but monotonicity allows the smaller disjoint suffix to be used. Consequently
\begin{equation}\label{eq:hybrid-pr-le-opt}
P_f(Y^\pm,R^\pm)\le\OPT_k(f,X_t).
\end{equation}

\emph{Capacity reduction for the mixed objective.} The preceding bound controls the error in terms of $f$. To control the benchmark, fix either pair $(Y,R)$ and let $U$ maximize $f_{\rm cmp}$ over sets of size at most $k$ in $X_t$. If $|U\cap Y|\le\kappa$, choose $O=U\cap Y$. Otherwise choose a uniform $\kappa$-subset $O$ of $U\cap Y$ and define $U_R=U\cap R$. Apply \eqref{eq:equal-marginal-sampling} to $I\mapsto f_{\rm cmp}(I\cup U_R)$ on $U\cap Y$, retaining its possibly nonzero base value:
\[
\E f_{\rm cmp}(O\cup U_R)\ge f_{\rm cmp}(U_R)+\frac{\kappa}{|U\cap Y|}\bigl(f_{\rm cmp}(U)-f_{\rm cmp}(U_R)\bigr)
\ge\frac\kappa k f_{\rm cmp}(U).
\]
In both cases some feasible $O$ satisfies $f_{\rm cmp}(O\cup R)\ge(\kappa/k)f_{\rm cmp}(U)$. Therefore
\begin{equation}\label{eq:hybrid-capacity-loss}
\max_{O\subseteq Y,\ |O|\le\kappa}f_{\rm cmp}(O\cup R)
\ge(1-2/B)\OPT_k(f_{\rm cmp},X_t).
\end{equation}

\emph{Averaging over replacement windows.} At this fixed time, the partial-replacement event depends only on the chosen window and has probability at most $1/B$, leaving total probability $s$ in $[1-1/B,1]$ for the disjoint before- and after-replacement events. Independence of the window and sampler allows \cref{eq:hybrid-static-interface,eq:hybrid-pr-le-opt,eq:hybrid-capacity-loss} to be applied to the appropriate complete sampled set on each event. Nonnegativity on the remaining event gives the first inequality below. The second uses $s\ge1-1/B$ for the nonnegative benchmark term and $s\le1$ for the subtracted error term:
\[
\begin{split}
\E f(S_t)&\ge s(1-2/B)\OPT_k(f_{\rm cmp},X_t)-s\delta_{\rm core}\OPT_k(f,X_t)\\
&\ge(1-1/B)(1-2/B)\OPT_k(f_{\rm cmp},X_t)-\delta_{\rm core}\OPT_k(f,X_t).
\end{split}
\]
Applying $\OPT_k(f_{\rm cmp},X_t)\le\OPT_k(f,X_t)$ then gives \cref{eq:hybrid-checkpoint-conclusion}. During the first block the whole prefix is maintained, and at the first transition $q=1$, the earlier prefix is empty and $R^-=X_t$, so its empty sampled set is exact. The boundary and incomplete-block conventions are those of \cref{lem:checkpoint}.
\end{proof}

\begin{proof}[Proof of \cref{thm:hybrid}]
We use direct recomputation for small capacities and checkpoints otherwise. For $0<\varepsilon<1$, choose $B=\lceil16/\varepsilon\rceil$ and $\eta=\delta_{\rm fail}=\varepsilon/64$, giving $\delta_{\rm core}=\varepsilon/32$. If $k<4B^2$, recompute the sampling distribution for the mixed objective with $\kappa=k$ and $R=\varnothing$ after each arrival. Its value satisfies \cref{eq:hybrid-goal}, with recourse at most $2k<8B^2$. For larger $k$, use uniform replacement windows independently of all sampler computations. Since $\delta_{\rm core}+3/B\le7\varepsilon/32<\varepsilon$, \cref{lem:hybrid-checkpoint} proves \cref{eq:hybrid-goal} with worst-case recourse at most $8B^2+2$. Appendix~\ref{app:hybrid-parameters} records the parameters and query bound.
\end{proof}

\subsection{The Curvature Threshold and Its Matching Lower Bound}

\begin{proof}[Proof of \cref{thm:curvature-main}]
For the upper bound, define $\ell(S)=(1-\vartheta)\sum_{i\in S}f(\{i\})$ and $g=f-\ell$ (see \citep[Lemma~6.2]{BNW25}). The offline decomposition in \citep[Section~6]{SVW17} uses terminal marginal weights $f(i\mid V\setminus\{i\})$, which depend on the full ground set. Our weights $\ell_i=(1-\vartheta)f(\{i\})$ are available at arrival. Subtracting $\ell$ preserves submodularity, and the curvature bound makes $g$ monotone on $V$. Thus \cref{thm:hybrid} applies, and $\sum_{i\in U}f(\{i\})\ge f(U)$ gives
\[
\ell(U)+\bconst g(U)
=\bconst f(U)+(1-\bconst)\ell(U)
\ge(1-\vartheta+\bconst\vartheta)f(U)=\rho_\vartheta f(U).
\]
This construction applies to every real $\vartheta\in[0,1]$. At $\vartheta=0$, maintaining the $k$ largest modular weights is exact and changes at most two elements per arrival.

For the lower bound, fix $\vartheta>0$, define $\lambda_{\rm mod}=(1-\vartheta)/\vartheta$, and take $f_0$ from \cref{thm:finite-lower}. The final element can have a very small marginal relative to its singleton value. Following the modular-augmentation method in \citep[Section~7 and Lemma~A.3]{SVW17}, we augment every element, including $r$, in proportion to its own singleton value:
\begin{equation}\label{eq:modular-hard}
f_{\lambda_{\rm mod}}(S)=f_0(S)+\lambda_{\rm mod}\sum_{i\in S}f_0(\{i\}).
\end{equation}

\emph{Curvature, including the final element.} Every marginal is at least $\lambda_{\rm mod}f_0(\{i\})$, while each singleton is $(1+\lambda_{\rm mod})f_0(\{i\})$. Hence every element of $X\cup\{r\}$ satisfies the curvature bound \eqref{eq:curvature-def} with $\vartheta=(1+\lambda_{\rm mod})^{-1}$.

\emph{Hiding.} For $k\ge64m$, current singletons lie in the flat band and share a value $\sigma_k\le5/k$, while the final singleton has value $c_*=c+1/(2m)$. Both are independent of $A$. Thus every current query $S\subseteq X$ acquires only the partition-independent term $\lambda_{\rm mod}\sigma_k|S|$, including when $|S|>k$.

\emph{A reference instance within the promise.} The reference instance in \cref{eq:curvature-reference} has current singleton values $(1+\lambda_{\rm mod})\sigma_k$, current-element marginals at least $\lambda_{\rm mod}\sigma_k$, and a modular final element. It therefore also has curvature at most $\vartheta$, so the algorithm's query bound and almost-sure termination apply to its execution. Its answers agree with the augmented hard instance on every balanced current query. The reference-execution argument of \cref{lem:transcript} therefore applies with this reference, using the same random tape and including the pre-arrival output in the union bound.

\emph{The ratio.} Define $\mathcal R(T)=w(1)$ and $\Delta_k=k\sigma_k+c_*$. The identity $v'(0)+c=1$ gives $1-1/(2k)\le\Delta_k\le1+65/m$ (\cref{eq:singleton-balance-finite}). On a good transcript, the repair argument bounds the original component by $f_0(S^+)\le\mathcal R(T)+38/m+5C/k$, even when the algorithm queries $f_{\lambda_{\rm mod}}$ and makes unlimited queries after the arrival. Here $S^+$ is its output after the final arrival. We bound the modular contribution separately. Granting it $r$ and $k$ old elements bounds this contribution by $\lambda_{\rm mod}\Delta_k$. The feasible comparison set $\{r\}\cup(A\setminus\{a_0\})$ has original value at least $1-5/k$ and modular value $\lambda_{\rm mod}(\Delta_k-\sigma_k)$, giving ratio at most
\begin{equation}\label{eq:curvature-exact-finite-ratio}
\frac{\mathcal R(T)+38/m+5C/k+\lambda_{\rm mod} \Delta_k}
{1-5/k+\lambda_{\rm mod}(\Delta_k-\sigma_k)}.
\end{equation}
Up to the finite corrections in $\Delta_k$ and $\sigma_k$, the augmentation adds $\lambda_{\rm mod}$ to both the algorithm's bound and the comparison value. The ratio is at most $(\mathcal R(T)+\lambda_{\rm mod})/(1+\lambda_{\rm mod})+130/m+(10C+12)/k$. Bad transcripts have probability at most $2(Q+1)e^{-k/(8192m^2)}$ and ratio at most one, so averaging over $A$ gives the fixed-instance bound in \cref{lem:finite-curvature}. Appendix~\ref{app:curvature-details} gives the numerical estimates.

For a fixed improvement $\zeta>0$, take $T=\sqrt2$, so $\mathcal R(T)=\bconst$, and choose an integer $m\ge32$ with $130/m<\zeta/4$, before choosing $k$ and $A$. Since $(\bconst+\lambda_{\rm mod})/(1+\lambda_{\rm mod})=\rho_\vartheta$, the finite bound forces $C=\Omega_{\vartheta,\zeta}(k)$ or $Q=\exp(\Omega_{\vartheta,\zeta}(k))$ for every real $\vartheta\in(0,1]$.
\end{proof}

\section{Conclusion and Further Questions}\label{sec:conclusion}

Polynomially many current value queries and worst-case constant recourse together limit the approximation to the supremum $\bconst=2-\sqrt2$. The obstruction combines information missing before an arrival with the cost of installing useful elements afterward. Coverage retains the offline $1-1/e$ guarantee up to $\varepsilon$, while the curvature law quantifies the loss as the objective approaches modularity.

The dependence on recourse remains to be determined more precisely. Our algorithm uses $O(\varepsilon^{-2})$ changes to approach $\bconst$ within $\varepsilon$. Must this budget diverge, and what is the optimal rate? At the other end, rerunning greedy uses at most $2k$ changes and achieves $1-1/e$. The intermediate tradeoff for a prescribed fraction of $k$ changes is not characterized by our bounds.

The query budget raises a related question. Our lower bound already excludes a fixed improvement with $e^{c_\zeta k}$ queries for a suitable $c_\zeta>0$, whereas unrestricted queries permit the $2/3$ threshold~\citep{DFL+25}. Determining the achievable approximation as a function of an exponential query budget would refine this separation beyond its two endpoints.

\appendix
\addtocontents{toc}{\protect\setcounter{tocdepth}{1}}
\section{Algorithmic Details}\label[appendix]{app:algorithms}

We give the checkpoint schedule, the current-value first-order routine, and the rounding argument used by the upper bounds. All-null tuples are stored implicitly until the first checkpoint. At that point $L$ elements have arrived and $k<B(L+1)$, so the tuples used by either checkpoint conversion have polynomial size in the current prefix and $B$.

\subsection{The Checkpoint Schedule}\label[appendix]{app:checkpoint-details}

The reduction in \cref{lem:checkpoint} is the Check-Point reduction of \citep[Theorem~3.2 and Section~4]{DFL+25}. Here we specify integer block and window lengths, bound core calls per arrival, and handle initialization, block boundaries, and a final incomplete block. The original convention counts insertions~\citep[Section~2]{DFL+25}, so resetting the recent set may remove $L$ elements without affecting that bound. For symmetric recourse, we instead maintain a lazy feasible superset of the target, bounding deletions by insertions. The value proof uses \cref{eq:robust-core-interface}, with each future suffix disjoint from its sample's current prefix and window choices independent of all samples.

\suppressfloats[t]
\begin{figure}[t]
\centering
\includegraphics[width=\linewidth]{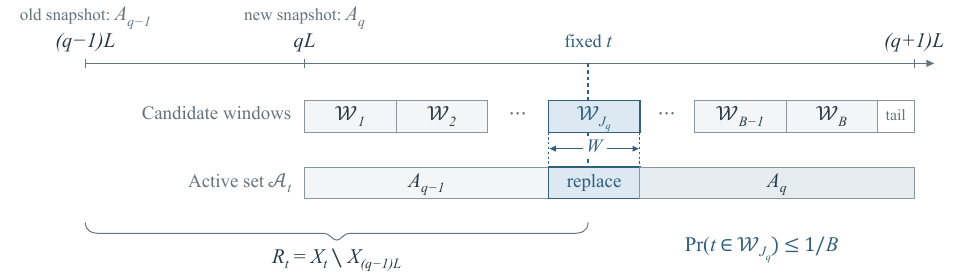}
\caption{Checkpointing within block $q$, with schematic horizontal spacing. The $B$ candidate windows are disjoint, so each fixed time belongs to at most one of them. Uniformly choosing $J_q$ bounds the probability of partial migration by $1/B$. The unused tail has length $L-BW\ge0$. Reserved capacity accommodates the recent set. The lazy feasible superset of $\mathcal A_t\cup R_t$ controls displayed-set recourse, including across block boundaries.}
\label{fig:checkpoint}
\end{figure}

\begin{lemma}[Lazy feasible superset]\label{lem:lazy}
Suppose $\mathcal T_0=S_0=\varnothing$, $\mathcal T_t\subseteq X_t$, $|\mathcal T_t|\le k$, and $|\mathcal T_t\setminus \mathcal T_{t-1}|\le d_{\rm add}$. In an insertion-only stream one can maintain $\mathcal T_t\subseteq S_t\subseteq X_t$, $|S_t|\le k$, and $|S_t\triangle S_{t-1}|\le2d_{\rm add}$ on every outcome.
\end{lemma}
\begin{proof}[Proof of \cref{lem:lazy}]
Let $I_t=\mathcal T_t\setminus S_{t-1}$ and first form the temporary set $U_t=S_{t-1}\cup I_t$, which is not output. Because $\mathcal T_{t-1}\subseteq S_{t-1}$,
\[
|I_t|\le |\mathcal T_t\setminus \mathcal T_{t-1}|\le d_{\rm add}.
\]
Set $d_t=\max\{0,|U_t|-k\}$. Since $\mathcal T_t\subseteq U_t$ and $|\mathcal T_t|\le k$, the set $U_t\setminus \mathcal T_t=S_{t-1}\setminus \mathcal T_t$ contains at least $d_t$ elements. Delete exactly $d_t$ of them in a fixed order and call the result $S_t$. Then $\mathcal T_t\subseteq S_t$, $|S_t|\le k$, and $d_t\le|I_t|$: before the insertions, $|S_{t-1}|\le k$, so the capacity excess cannot exceed the number inserted. Therefore
\[
|S_t\triangle S_{t-1}|=|I_t|+d_t\le2d_{\rm add}.
\]
When the target only shrinks, $I_t=\varnothing$ and no deletion is performed. All elements of the maintained superset belong to the insertion-only prefix $X_t$.
\end{proof}

\begin{proof}[Proof of \cref{lem:checkpoint}]
\emph{Parameters.} If $k<4B^2$, recompute ordinary greedy after every arrival. Its approximation is at least $1-\nicefrac{1}{e}$ and its worst-case recourse is at most $2k<8B^2$. Otherwise, define
\[
L=\lfloor k/B\rfloor,\qquad \kappa=k-2L,\qquad
W=\lfloor L/B\rfloor,\qquad c_{\rm step}=\lceil\kappa/W\rceil.
\]
Here $L$ is the block length, $\kappa$ is the sample capacity, $W$ is the replacement-window length, and $c_{\rm step}$ bounds the number of positions replaced per update. These integers satisfy $L\ge4B$, $W\ge1$, and $BW\le L$. Since $B$ is an integer and $k\ge4B^2$, we have $W=\lfloor k/B^2\rfloor\ge k/(2B^2)$, so $c_{\rm step}\le2B^2$. Number the sampled set's positions from $1$ to $\kappa$, padding with nulls and ordering real elements by a fixed input order. Let $A_0$ be the all-null tuple on $X_0=\varnothing$. In set expressions, $A_q$ denotes the set of its distinct nonnull entries.

\emph{The update rule.} Computation on the current prefix follows the maintained-set update, so the new tuple is first used at the next arrival. Initialize $S_0=\varnothing$ and $A_0$ as above. At time $t$, perform the following operations in order.
\begin{enumerate}[label=\arabic*.]
\item If $t\le L$, set $\mathcal T_t=S_t=X_t$.
\item If $t>L$, express $t=qL+\iota$ with $q\ge1$ and $1\le \iota\le L$. When $\iota=1$, draw $J_q$ uniformly from $\{1,\ldots,B\}$. Define
\[
a_q=(J_q-1)W,\qquad j_t=\min\{\kappa,c_{\rm step}\max\{0,\iota-a_q\}\}.
\]
The active set consists of the distinct nonnull entries
\[
\mathcal A_t=\bigl(\{A_q[i]:1\le i\le j_t\}\cup
\{A_{q-1}[i]:j_t<i\le\kappa\}\bigr)\setminus\{\bot\}.
\]
Set $R_t=X_t\setminus X_{(q-1)L}$ and $\mathcal T_t=\mathcal A_t\cup R_t$, then update $S_t$ by the superset rule of \cref{lem:lazy}.
\item If $t$ is a multiple of $L$, compute a fresh sampled tuple $A_{t/L}$ on $X_t$ for use in the next block.
\end{enumerate}
The sampled tuples and window choices use independent randomness. The value of $J_q$ is sampled once and kept fixed throughout its block. Only step~3 calls the robust core, at most once per arrival. The remaining updates use no value queries.

\emph{Feasibility and recourse.} Each target uses at most $\kappa$ tuple positions and $2L$ recent elements, so $|\mathcal T_t|\le k$. We show that $|\mathcal T_t\setminus\mathcal T_{t-1}|\le c_{\rm step}+1$. During initialization the target gains only the new arrival. Inside a later block, at most $c_{\rm step}$ positions change and the recent suffix gains one arrival. Since $c_{\rm step}W\ge\kappa$ and $BW\le L$, replacement finishes by the end of the chosen window and hence by the block's end. At a block boundary, this completed tuple becomes the old tuple for the next block. Its first update replaces at most $c_{\rm step}$ positions, while the recent set only discards old elements and adds the new arrival. At the first transition, $R_{L+1}=X_{L+1}$ and $\mathcal T_{L+1}=X_{L+1}$, so the same insertion bound holds.

By \cref{lem:lazy}, the update maintains $\mathcal T_t\subseteq S_t\subseteq X_t$ and $|S_t|\le k$, with large-$k$ recourse at most $2(c_{\rm step}+1)\le4B^2+2$. Together with the small-$k$ branch, this gives the uniform bound $8B^2+2$, including recent-set resets.

\emph{Value at a fixed time.} In this branch, $S_t=\mathcal T_t=X_t$ for $t\le2L$: step~1 covers $t\le L$, and $R_t=X_t$ when $L<t\le2L$. Fix $t=qL+\iota$ with $q\ge2$ and $1\le\iota\le L$. Since the adversary is oblivious, the suffixes below are fixed independently of all samples. Define
\[
\mathcal E_{\rm old}=\{j_t=0\},\qquad
\mathcal E_{\rm new}=\{j_t=\kappa\}.
\]
These events are disjoint and depend only on $J_q$, independently of both samples. The remaining event satisfies
$\{0<j_t<\kappa\}\subseteq\{a_q<\iota\le a_q+W\}$.
For fixed $\iota$, at most one of the $B$ disjoint windows contains it, so the remaining event has probability at most $1/B$. It is empty in the unused tail $BW<\iota\le L$.

For the old sample, define
\[
R_t^{\rm old}=X_t\setminus X_{(q-1)L},
\qquad
h_t^{\rm old}(S)=f(S\cup R_t^{\rm old})\quad(S\subseteq X_{(q-1)L}).
\]
This is a legal future fixed independently of the randomness used for $A_{q-1}$. On $\mathcal E_{\rm old}$, the target is exactly $A_{q-1}\cup R_t^{\rm old}$. For every $U\subseteq X_t$ of size at most $\kappa$, the set $U\cap X_{(q-1)L}$ is a feasible comparison set and
\[
h_t^{\rm old}(U\cap X_{(q-1)L})
=f((U\cap X_{(q-1)L})\cup R_t^{\rm old})\ge f(U).
\]
Thus its benchmark is at least $\OPT_\kappa(X_t)$.

For the new sample, define
\[
R_t^{\rm new}=X_t\setminus X_{qL},
\qquad
h_t^{\rm new}(S)=f(S\cup R_t^{\rm new})\quad(S\subseteq X_{qL}).
\]
This legal future is fixed independently of $A_q$, and the same comparison-set argument shows that its benchmark is at least $\OPT_\kappa(X_t)$. On $\mathcal E_{\rm new}$, the target is $A_q\cup R_t^{\rm old}$, which contains $A_q\cup R_t^{\rm new}$. The containment $\mathcal T_t\subseteq S_t$, monotonicity, and independence of the window and samples therefore give
\[
\begin{aligned}
\E f(S_t)
&\ge \Prob(\mathcal E_{\rm old})\E h_t^{\rm old}(A_{q-1})
 +\Prob(\mathcal E_{\rm new})\E h_t^{\rm new}(A_q)\\
&\ge\bigl(\Prob(\mathcal E_{\rm old})+\Prob(\mathcal E_{\rm new})\bigr)
\alpha\OPT_\kappa(X_t)\\
&\ge(1-1/B)\alpha\OPT_\kappa(X_t).
\end{aligned}
\]
The first inequality uses nonnegativity on the remaining event.

Finally, a uniform $\kappa$-subset of an optimal $k$-set, padded by nulls if necessary, has expected value at least $(\kappa/k)\OPT_k(X_t)$ by \eqref{eq:equal-marginal-sampling}. Hence $\OPT_\kappa(X_t)\ge(1-2/B)\OPT_k(X_t)$, giving the claimed approximation. The same fixed-time argument applies to a final incomplete block.
\end{proof}

\subsection{Computing a First-Order Gap}\label[appendix]{app:stationarity}

We use a fixed-step implementation of the standard Frank--Wolfe gap analysis for smooth nonconcave objectives \citep{LJ16}. An empty current set returns the empty vector without queries. Otherwise let $X$ be the current set, $n=|X|$, and $M=\max_{i\in X}f(\{i\})$ for its normalized monotone submodular objective $f$. The routine applies to $J=\Phi_\Lambda^f/(1+\Lambda)$ with $1\le\Lambda\le\sqrt2$, the coverage potential $J=H_\varnothing^f$, and the mixed potential $J=\Xi$ for $f=g+\ell$ with the access specified in \cref{thm:hybrid}. In the mixed case, $M$ is defined using $f$, not just $g$. The effective capacity $\min\{\kappa,n\}$ defines the same domain $\mathcal P_\kappa$.

\begin{lemma}[First-order routine]\label{lem:first-order}
For $n\ge1$, any of these potentials, any integer $\kappa\ge1$, and any $\eta,\delta_{\rm fail}\in(0,1)$, the routine below returns $\boldsymbol{x}\in\mathcal P_\kappa$ on every outcome. With probability at least $1-\delta_{\rm fail}$, it satisfies
\[
\max_{\boldsymbol{y}\in\mathcal P_\kappa}
\ip{\nabla J(\boldsymbol{x})}{\boldsymbol{y}-\boldsymbol{x}}\le\eta M.
\]
It uses $O(n\kappa^5\eta^{-4}\log(n\kappa/(\eta\delta_{\rm fail})))$ current value queries, without querying the potential or any future.
\end{lemma}
\begin{proof}
\emph{Bounds and degenerate cases.} On $\mathcal P_\kappa$, the needed analytic bounds are
\begin{equation}\label{eq:smooth-bounds}
0\le J(\boldsymbol{x})\le2\kappa M,\qquad
0\le\partial_iJ(\boldsymbol{x})\le2M,\qquad
|\partial_{ij}J(\boldsymbol{x})|\le M.
\end{equation}
The gradient bounds follow from bounded current marginals. Since $J(0)=0$, integrating the gradient along the segment from $0$ to $\boldsymbol{x}$ gives the range bound. For a Poisson response, diagonal Hessian entries are negatives of gradient coordinates, and mixed entries are weighted expectations of discrete second differences. Both have magnitude at most $M$. For $\Phi_\Lambda^f$, integrating $t\partial_{ij}H_\varnothing^f(t\boldsymbol{x})$ preserves this bound because $\Lambda^2/2\le1$. These bounds also cover unnormalized $\Phi_\Lambda^f$. For $\Xi$, the modular term contributes no Hessian, and $\partial_i\Xi\le\ell_i+g(\{i\})=f(\{i\})$, giving the sharper range bound $\Xi(\boldsymbol{x})\le\kappa M$.

If $M=0$, every current marginal is zero, so return $\boldsymbol{x}=0$ with gap zero. If $n\le\kappa$, use $\boldsymbol{x}=\one$, which has zero first-order gap because all gradient coordinates are nonnegative.

\emph{The routine.} Assume $n>\kappa$ and $M>0$. Define
\[
a=\frac{\eta}{16\kappa},\qquad
\eta_{\rm step}=\frac{\eta}{16\kappa^2},\qquad
I=\left\lceil\frac{128\kappa^3}{\eta^2}\right\rceil+1.
\]
Start at $x_i=\kappa/n$. At each iteration estimate the gradient by $\widehat{\boldsymbol{g}}$ using the nonnegative marginal samples described below. Let $\boldsymbol{y}$ indicate its top $\kappa$ coordinates, with deterministic tie-breaking. Nonnegativity makes $\boldsymbol{y}$ a maximizer of the estimated linear objective over $\mathcal P_\kappa$. If $\ip{\widehat{\boldsymbol{g}}}{\boldsymbol{y}-\boldsymbol{x}}\le\eta M/2$, return $\boldsymbol{x}$. Otherwise update $\boldsymbol{x}\leftarrow\boldsymbol{x}+\eta_{\rm step}(\boldsymbol{y}-\boldsymbol{x})$. Run at most $I$ iterations, returning the current feasible point if the stopping test has not succeeded.

\emph{Progress on the accuracy event.} Suppose every estimate used has infinity-norm error at most $aM$. Every feasible direction has $\ell_1$ norm at most $2\kappa$, so for all $\boldsymbol z\in\mathcal P_\kappa$,
\[
\bigl|\ip{\widehat{\boldsymbol g}-\nabla J(\boldsymbol x)}{\boldsymbol z-\boldsymbol x}\bigr|
\le2\kappa aM=\eta M/8.
\]
At stopping, maximality of $\boldsymbol y$ for the estimated objective therefore bounds the true gap by $\eta M/2+\eta M/8=5\eta M/8<\eta M$.

Otherwise, the true directional derivative along $\boldsymbol y-\boldsymbol x$ is greater than $3\eta M/8$. The Hessian bound in \cref{eq:smooth-bounds} gives an improvement of at least
\[
\eta_{\rm step}\frac{3\eta M}{8}-2\kappa^2M\eta_{\rm step}^2
=\frac{\eta^2M}{64\kappa^2}.
\]
The range bound $J\le2\kappa M$ excludes $I$ nonstopping iterations, so the routine stops within its budget on the accuracy event.

\emph{Gradient estimates and query bound.} For a scale potential built from the current objective $g$, draw $t$ uniformly from $[0,\Lambda]$ and sample the Poisson set $Z_{t\boldsymbol{x}}$. Since $g(i\mid Z)=0$ when $i\in Z$, the identity
\[
\partial_i\Phi_\Lambda^g(\boldsymbol{x})
=\Lambda\E_{t,Z_{t\boldsymbol{x}}}g(i\mid Z_{t\boldsymbol{x}})
\]
gives an unbiased estimator using two current value queries. This form incorporates the factor $e^{-tx_i}$ through the event that $i$ is absent. For the normalized potential, divide the sample by $1+\Lambda$. For the mixed objective, also add $\ell_i$. For coverage, use $f(i\mid Z_{\boldsymbol{x}})$ directly, whose expectation is \cref{eq:coverage-gradient}. All these samples lie in $[0,2M]$.

Average
\[
N_0=\left\lceil\frac2{a^2}\log\frac{2nI}{\delta_{\rm fail}}\right\rceil
\]
independent samples for each coordinate, using fresh randomness at every iteration. Conditional on the preceding optimization history, Hoeffding's inequality \citep{Hoe63} bounds the probability of error greater than $aM$ by $\delta_{\rm fail}/(nI)$. A union bound over at most $nI$ estimates gives simultaneous accuracy with probability at least $1-\delta_{\rm fail}$, proving the gap guarantee.

The number of current value queries is at most $1+n+2nIN_0$, namely
\begin{equation}\label{eq:oracle-work}
O\left(n\kappa^5\eta^{-4}\log\frac{n\kappa}{\eta\delta_{\rm fail}}\right).
\end{equation}
Generating query sets and selecting the top $\kappa$ coordinates require polynomially many arithmetic and comparison operations, in addition to the specified exact draws. Every iterate remains feasible, including when an estimate is inaccurate.
\end{proof}

\paragraph{Prices at a prescribed point.}\label{par:prescribed-prices}
For $\Lambda\in[1,\sqrt2]$, the same sampling method estimates $\boldsymbol{p}_\Lambda(\boldsymbol{x})=\nabla\Phi_\Lambda^f(\boldsymbol{x})/(1+\Lambda)$ at any prescribed $\boldsymbol{x}\in[0,1]^n$, without a stationarity search. Fix $\xi>0$ and $\delta_{\rm fail}\in(0,1)$. When $M=0$, return zero prices. Otherwise take $a=\min\{1/2,\xi/(4n)\}$ and use the sample budget above with $I=1$. The scaled marginal samples remain bounded even if $\sum_i x_i>\kappa$, so the same concentration bound gives simultaneous coordinate error at most $aM$.

Clip the result to the price box $\prod_i[0,M_i]$, which contains the exact price vector. On the accuracy event,
\[
\bigl|\ip{\widehat{\boldsymbol{p}}-\boldsymbol{p}_\Lambda(\boldsymbol{x})}{\boldsymbol{x}-\one_O}\bigr|
\le naM\le\xi M
\qquad(O\subseteq X).
\]
Thus one event controls every comparison set and future constraint. At most $1+n+2nN_0$ current value queries are used, polynomial in $n$, $1/\xi$, and $\log(1/\delta_{\rm fail})$. Fresh randomness makes the guarantee conditional on the preceding history when the prescribed point is chosen adaptively. This proves \cref{cor:computable-prices}.

\Needspace{8\baselineskip}
\subsection{Randomized Pipage Rounding}\label[appendix]{app:rounding}

Pipage rounding was introduced in \citep{AS04}. We use its randomized cardinality specialization as described in \citep[Appendix~A]{CVZ10}.

\begin{lemma}[Randomized pipage rounding]\label{lem:pipage}
Let $\kappa\ge1$ be an integer and let $\boldsymbol{x}\in\mathcal P_\kappa(X)$. Using only $\boldsymbol{x}$ and $\kappa$ and no value queries, the routine below returns $A\subseteq X$ with $|A|\le\kappa$ on every outcome and $\Prob(i\in A)=x_i$ for every $i\in X$. For every fixed submodular function $h:2^X\to\R$,
\[
\E h(A)\ge F_h(\boldsymbol{x}),
\]
where $F_h$ is its multilinear extension. In particular, for every modular function $\ell(S)=\sum_{i\in S}\ell_i$, we have $\E\ell(A)=\boldsymbol\ell\cdot\boldsymbol x$.
\end{lemma}
\begin{proof}
For $n=|X|>0$, define $\kappa_{\rm eff}=\min\{\kappa,n\}$. Fill the deficit $\kappa_{\rm eff}-\sum_i x_i$ with full null coordinates and at most one fractional null. Since $\kappa_{\rm eff}\le n$, at most $n$ nulls suffice. The extended objective ignores them. The empty instance simply returns the empty set.

Choose two fractional coordinates $i,j$ by a fixed rule, and define
\[
t_+=\min\{1-x_i,x_j\},\qquad
t_-=\min\{x_i,1-x_j\},\qquad
p=\frac{t_-}{t_++t_-}.
\]
Move to $\boldsymbol{x}+t_+(e_i-e_j)$ with probability $p$ and to $\boldsymbol{x}-t_-(e_i-e_j)$ otherwise. Both endpoints remain in the cube, preserve total mass, and fix at least one fractional coordinate. Their mean is $\boldsymbol{x}$. Integral total mass rules out a final single fractional coordinate, so at most $2n$ steps yield a feasible set of real elements. Each coordinate preserves its conditional expectation at every exchange, giving $\Prob(i\in A)=x_i$ and the modular-value identity.

For any fixed submodular function $h$, the multilinear extension $F_h$ has exchange-direction second derivative $-2\partial_{ij}F_h\ge0$. Convexity and the mean-preserving endpoint choice imply $\E[F_h(\boldsymbol{x}^{\rm next})\mid\boldsymbol{x}]\ge F_h(\boldsymbol{x})$. Iterating gives $\E h(A)\ge F_h(\boldsymbol{x})$. This argument does not require $h(\varnothing)=0$, so it also applies to $h(S)=f(S\cup R)$. The rounding law is independent of $h$ and is the same for every fixed future.
\end{proof}

\section{Hard-Instance Construction}\label[appendix]{app:profiles}\label[appendix]{app:profile-construction}

We verify the continuous profiles in \cref{lem:flat} and the induced discrete oracle in \cref{lem:continuous-to-discrete}. These results, together with the properties of $\psi$ in \cref{app:transcript}, prove parts~(i)--(iii) of \cref{lem:hard-profiles}, and the unit-mass evaluations below prove part~(iv). The final subsection proves the fixed-query concentration bound. The flattening and common-regularization principles follow \citep[proof of Lemma~3.2]{Von13}. We verify the additional cross-branch comparisons required by the final element.

\subsection{Origin of the Profiles}\label[appendix]{app:hkmy-normalization}

This subsection identifies the inherited initial segment. The subsequent verification uses the explicit profiles in \cref{eq:vw,eq:FK,eq:clamped} directly.
Let $f_{\rm HK}(b,a,p)$ be the finite function in \citep[Definition~2.3]{HKMY22}. Its integer parameters are the capacity $K_{\rm HK}\ge2$ and $h_{\rm HK}\ge K_{\rm HK}$, with $L=K_{\rm HK}-1$ red elements and one purple element. Here $b,a,p$ are the blue, red, and purple counts, with $a\in\{0,\ldots,L\}$ and $p\in\{0,1\}$. Normalize by $M_L=f_{\rm HK}(0,L,1)=L(L+h_{\rm HK})+h_{\rm HK}(h_{\rm HK}+1)/2$, as in their Lemma~2.1. For the blue counts $0\le b\le h_{\rm HK}$ used here, the defining marginals give
\[
\begin{aligned}
f_{\rm HK}(b,a,0)&=b(L+h_{\rm HK})-b(b-1)/2+a(L+h_{\rm HK}-b),\\
f_{\rm HK}(b,a,1)&=h_{\rm HK}(h_{\rm HK}+1)/2+Lb+a(L+h_{\rm HK}-b).
\end{aligned}
\]
As $L\to\infty$ with $h_{\rm HK}/L\to T$, $b/L\to t\le T$, and $a/L\to\eta$, division by $M_L$ gives $v(t)+\eta v'(t)$ and $w(t)+\eta v'(t)$ from \eqref{eq:vw}. This identifies only the normalized continuous initial segment. Our balanced coordinates, smooth continuation, exponential regularizer, clipping, and future correction are verified below.

\subsection{The Continuous Profiles}\label[appendix]{app:profile-derivatives}

Use the profiles $v,w$ and constants $d,c$ from \cref{eq:vw}, with $T\in[\nicefrac{7}{5},\nicefrac{3}{2}]$ and integer $m\ge32$. On $[0,1]\times[0,m]$, set
\[
\omega_m=1+1/m,\quad t=\omega_m y,\quad u=x-y/m,\quad s=x+y=t+u,
\quad\delta=1/(64m).
\]
For $z=t+\clip(u,-\delta,\delta)$ and $q=s-z$, the tangency point $z$ lies between $t$ and $s$, and
\[
|q|\le|u|\le1,\qquad |z-t|\le\delta,\qquad |q-u|\le\delta.
\]
The common regularizer $\phi_m(s)=32(1-e^{-s})/m$ ensures the required derivative margins within each branch. Clipping makes the current branch depend only on total mass within the balanced band, but can reverse the gradient ordering across branches. The correction $32\delta e^{-s}$ to the future branch restores this ordering without changing current oracle answers.

\begin{lemma}[Compatible exact flattening]\label{lem:flat}
On $[0,1]\times[0,m]$, the profiles in \cref{eq:clamped} have the following properties.
\begin{enumerate}[label=\textup{(\alph*)},leftmargin=2em]
\item \emph{Values and exact hiding.} Both profiles are nonnegative, $F_\delta(0,0)=0$, and $F_\delta(x,y)=v(s)+\phi_m(s)$ whenever $|u|\le\delta$.
\item \emph{Within each branch.} Each profile $P\in\{F_\delta,K_\delta\}$ is $C^1$ with locally Lipschitz gradient. For $i,j\in\{x,y\}$,
\[
16e^{-s}/m\le\partial_iP\le4,\qquad
\partial_{ij}P\le-16e^{-s}/m\quad\text{almost everywhere}.
\]
Each first partial derivative is nonincreasing as either coordinate increases, along every coordinate segment in the rectangle, including segments crossing or lying on piece boundaries.
\item \emph{Across branches.} Pointwise on the rectangle,
\begin{equation}\label{eq:flat-properties}
K_\delta\ge F_\delta,\qquad
\partial_iF_\delta-\partial_iK_\delta\ge16\delta e^{-s}
\quad(i=x,y).
\end{equation}
\end{enumerate}
\end{lemma}
\begin{proof}
\emph{Values.} For a concave differentiable function, the tangent value $v(z)+(s-z)v'(z)$ decreases as $z$ moves toward $s$. Thus
\[
v(s)+\phi_m(s)\le F_\delta\le F\le K\le K_\delta,
\]
where $K-F=w(t)-v(t)\ge0$ by \cref{eq:profile-identities}. In the inner band, $z=s$ and $q=0$, giving exact hiding. Nonnegativity and $F_\delta(0,0)=0$ follow immediately.

\emph{Derivatives before regularization.} On the quadratic piece, $v'=d(T+1-z)$, $v''=-d$, and $v'''=0$. On the tail, $v'=-v''=v'''=de^{T-z}$. The first two functions below are increasing in $T\in[0,3/2]$. Using $e^{5/2}<49/4$, $d(3/2)=8/29$, and $\omega_m\le33/32$ gives
\[
de^{T+1}<4,\qquad d(T+1)e^{T+1}<10,\qquad
2\omega_m{}de^{T+1}<9.
\]
Outside the band, $\clip(u,-\delta,\delta)$ is constant. For $J=v(z)+qv'(z)$, differentiation gives
\begin{equation}\label{eq:flat-derivatives}
\begin{aligned}
J_x&=v'(z),&
J_y&=v'(z)+\omega_m{}qv''(z),\\
J_{xx}&=0,&
J_{xy}&=\omega_m{}v''(z),\\
J_{yy}&=\omega_m(1-1/m)v''(z)+\omega_m^2qv'''(z).
\end{aligned}
\end{equation}
On the quadratic piece, $z\le T$ and $q\le1$ give
$J_y=d(T+1-z-\omega_m q)\ge-d/m$ and
$J_{yy}=-\omega_m(1-1/m)d\le0$.
On the tail, $v'(z)\le4e^{-s}$ because $s-z=q\le1$, so
\[
J_y=v'(z)(1-\omega_m q)\ge-4e^{-s}/m,\qquad
J_{yy}=\omega_m v'(z)(\omega_m q-1+1/m)<9e^{-s}/m.
\]
The remaining derivatives already have the required signs. In the inner band, $J=v(s)$, with positive gradients and nonpositive Hessian entries.

For the future branch before regularization, $K_0=w(t)+uv'(t)$ satisfies
\begin{equation}\label{eq:K-derivatives}
\begin{aligned}
(K_0)_x&=v'(t),&
(K_0)_y&=\omega_m{}w'(t)-v'(t)/m+\omega_m{}u v''(t),\\
(K_0)_{xx}&=0,&
(K_0)_{xy}&=\omega_m{}v''(t),\\
(K_0)_{yy}&=\omega_m^2w''(t)-2\omega_m{}v''(t)/m+\omega_m^2u v'''(t).
\end{aligned}
\end{equation}
On the quadratic piece, $w'=d$ and $w''=0$, so
$(K_0)_y\ge-d(T+1)/m$ and $(K_0)_{yy}=2\omega_m d/m$.
On the tail, the preceding current-branch estimates apply with $z=t$ and $q=u$.
Whenever a quadratic piece is evaluated, $s\le T+1$. The numerical bounds above therefore show that, in both unregularized branches, every coordinate gradient is at least $-10e^{-s}/m$ and every Hessian entry is at most $9e^{-s}/m$ wherever it exists.

\emph{Regularization margins and upper bounds.} The common regularizer must cover these sign violations and leave room for the future correction. It adds $32e^{-s}/m$ to every gradient and $-32e^{-s}/m$ to every Hessian entry. Hence $F_\delta$ has lower gradient margin $22e^{-s}/m$ and upper Hessian bound $-23e^{-s}/m$. The future correction changes these by $-32\delta e^{-s}=-e^{-s}/(2m)$ and $e^{-s}/(2m)$, respectively. The corresponding future bounds are $43e^{-s}/(2m)$ and $-45e^{-s}/(2m)$, which imply the stated margins of $16e^{-s}/m$.
For the upper bounds, $v',w',|v''|\le1$, $|u|,|q|\le1$, and $32/m\le1$. The displayed formulas bound the current $y$ derivative by $2+\omega_m<4$ and the future $y$ derivative by $1+2\omega_m<4$. Both $x$ derivatives and the current inner-band derivatives are at most two. The future correction only decreases its gradients.

\emph{Cross-branch gradients.} The identity $K-F=w(t)-v(t)$ gives $\partial_xK=\partial_xF$ and $\partial_yK\le\partial_yF$. Thus the first comparison has no slack before flattening. To bound its perturbation, every point between $t$ and $s$ is within one of $s$, so the scalar formulas above give $|v''|\le4e^{-s}$ and $|v'''|\le4e^{-s}$ almost everywhere on that interval. The derivative $v''$ is continuous at the profile junction, so these bounds may be integrated across it.
Outside the band, cancellation of the common regularizer gives
\[
\begin{aligned}
\partial_xF_\delta-\partial_xF
 &=v'(z)-v'(t),\\
\partial_yF_\delta-\partial_yF
 &=v'(z)-v'(t)
   +\omega_m\bigl((q-u)v''(z)+u(v''(z)-v''(t))\bigr).
\end{aligned}
\]
The elementary bounds on $q,u,z$ imply
\[
|\partial_xF_\delta-\partial_xF|\le4\delta e^{-s},\qquad
|\partial_yF_\delta-\partial_yF|
\le4(1+2\omega_m)\delta e^{-s}<16\delta e^{-s}.
\]
Inside the band, the two differences are $v'(s)-v'(t)$ and
$v'(s)-v'(t)-\omega_m u v''(t)$, respectively. Since $|u|\le\delta$, both are bounded in absolute value by $4(1+\omega_m)\delta e^{-s}$.
This proves \cref{eq:gradient-perturbation}. Consequently, for $i=x,y$,
\[
\partial_iF_\delta-\partial_iK_\delta
=(\partial_iF_\delta-\partial_iF)
 +(\partial_iF-\partial_iK)+32\delta e^{-s}
\ge16\delta e^{-s}.
\]

\emph{Regularity and coordinate segments.} At a clipping boundary, $q=0$, so the outer gradient $(v'(z),v'(z)+\omega_m qv''(z))$ agrees with the inner gradient $(v'(s),v'(s))$. At the profile junction, $v',v'',w'$ are continuous. The regularizer and future correction are smooth. The profiles are therefore $C^1$, and their gradients are locally Lipschitz because the piecewise Hessians are bounded and the gradients match at the finitely many interfaces.
A coordinate segment not contained in a piece boundary splits into finitely many smooth pieces. Integrating the Hessian inequalities on each piece and joining the continuous endpoint values shows that each first partial derivative is nonincreasing along the whole segment. A segment lying in a boundary is a limit of parallel interior segments, so the same conclusion follows by gradient continuity. These comparisons hold along every coordinate segment and can therefore be integrated to give the discrete marginal inequalities below.
\end{proof}

\emph{Values used in Section~\ref{sec:lower}.} At the balanced unit-mass point, $u=0$ and $t=1$, whereas the hidden point has $(x,y)=(1,0)$. Direct substitution gives
\[
\begin{aligned}
K_\delta\left(\frac1{m+1},\frac m{m+1}\right)
 &=\mathcal R(T)+\phi_m(1)+32\delta e^{-1}
 \le\mathcal R(T)+\frac{32}{m}+32\delta,\\
K_\delta(1,0)&=1+\phi_m(1)+32\delta e^{-1}\ge1.
\end{aligned}
\]
In particular,
\[
K_\delta(1,0)-K_\delta\left(\frac1{m+1},\frac m{m+1}\right)
=1-\mathcal R(T).
\]
The common additions preserve this absolute difference, not the ratio. These profile evaluations need not correspond to feasible sets: the hidden point includes $r$ and has $k+1$ elements, and the balanced point may lie off the count grid. The comparison with a feasible set is made in the proof of \cref{thm:finite-lower}.
Also $K_\delta(0,0)=c+32\delta$. If $k\ge64m$, both types of current singletons lie in the flat band and have the common value
\begin{equation}\label{eq:hard-singleton}
\sigma_k=v(1/k)+\phi_m(1/k),\qquad 0\le\sigma_k\le4/k\le5/k.
\end{equation}
The upper bound follows by integrating a current coordinate derivative from the origin.

\subsection{The Discrete Oracle}\label[appendix]{app:discrete-oracle}

We adapt the continuous-to-discrete criterion of \citep[Lemma~3.1]{Von13} to the two count coordinates and verify the additional comparisons involving $r$.

\begin{lemma}[Continuous signs certify the full discrete oracle]
\label{lem:continuous-to-discrete}
The function in \cref{eq:hidden-oracle} is normalized, nonnegative,
monotone, and submodular on all of $2^{X\cup\{r\}}$, including sets larger than the maintained capacity. Every current-element marginal, against any background set in either branch, lies in $[0,5/k]$.
\end{lemma}
\begin{proof}
\emph{Monotonicity and small marginals.} For $P\in\{F_\delta,K_\delta\}$, the marginal of a current element in $A$ is the integral of $\partial_xP$ over an interval of length $1/k$, whereas an element in $\bar A$ uses $\partial_yP$ instead. By \cref{lem:flat}(b), these marginals lie in $[0,4/k]\subseteq[0,5/k]$. The marginal of $r$ is $D=K_\delta-F_\delta\ge0$ by part~(c), so $f_A$ is monotone. It is normalized because $F_\delta(0,0)=0$, and hence nonnegative.

\emph{Submodularity.} We check pairs of distinct elements, grouped by type. Within either branch, part~(b) shows that an $A$-marginal decreases when either $x$ or $y$ increases, covering $(A,A)$ and $(A,\bar A)$. The $y$-direction comparison for a $\bar A$-marginal covers $(\bar A,\bar A)$. The diagonal Hessian signs are needed for distinct elements of the same type. For $(A,r)$ and $(\bar A,r)$, part~(c) shows that adding $r$ cannot increase the integral of the corresponding coordinate derivative. These five cases exhaust all pairs, since there is only one $r$.
The within-branch comparisons hold on entire coordinate segments by part~(b), including segments crossing or lying on piece boundaries, while part~(c) holds pointwise. Integrating over the relevant length-$1/k$ intervals therefore gives the required discrete diminishing-return inequalities on the full count grid.
\end{proof}

\subsection{Concentration and Parameters}\label[appendix]{app:transcript}

Choose $A$ uniformly among the $k$-subsets of $X$, where $|X|=(m+1)k$. For any fixed $S\subseteq X$ independent of $A$, the variable $Z=|A\cap S|$ counts marked identifiers in a sample of size $k$ without replacement, and
\[
\E Z=\frac{|S|}{m+1},\qquad
u_A(S)=\frac{m+1}{m}\left(\frac Zk-\frac{\E Z}{k}\right).
\]
Thus $|u_A(S)|>\delta$ is equivalent to
$|Z/k-\E Z/k|>m\delta/(m+1)$.
Hoeffding's inequality for sampling without replacement \citep[Section~6]{Hoe63} gives
\[
\Prob\{|u_A(S)|>\delta\}
\le2\exp\left(-\frac{2m^2\delta^2}{(m+1)^2}k\right)
\le2e^{-\delta^2k/2},
\]
proving \cref{eq:hypergeo}. No restriction on $|S|$ is used, so this fixed-set bound covers arbitrary current queries, including sets larger than $k$. Its application to adaptive queries and the possibly unqueried output is proved in \cref{lem:transcript}.

\emph{The reference profile.} The function $\psi=v+\phi_m$ is nonnegative, nondecreasing, and concave on $[0,m+1]$, with $\psi(0)=0$, because both summands have these properties. Hence \cref{eq:reference-oracle} defines a normalized monotone submodular function, and $r$ has zero marginal.

\emph{Parameter order.} For a target gap $\zeta>0$, first set $T=\sqrt2$ and fix $m$ as a function of $\zeta$, then take $k$ sufficiently large, as in the deduction of \cref{thm:main}(ii) from \cref{thm:finite-lower}. Since $|X\cup\{r\}|=(m+1)k+1=\Theta_\zeta(k)$, the resulting query threshold is exponential in the total ground-set size for each fixed $\zeta$.

\section{Computational Hardness for Coverage}\label[appendix]{app:coverage-details}

\begin{proposition}[Computational optimality for coverage]\label{prop:coverage-hardness}
Fix a constant $0<\delta<1/e$. Suppose a randomized polynomial-time online algorithm in the standard Turing model, on every explicitly represented nonnegative weighted coverage instance with final ground set $X$, returns at the final arrival a subset of $X$ of size at most $k$ whose expected value is at least $(1-1/e+\delta)\OPT_k(X)$. Then $\mathsf{NP}\subseteq\mathsf{BPP}$. No recourse restriction is imposed.
\end{proposition}
\begin{proof}
Let $\aconst=1-1/e$. The gap reduction in the proof of \citep[Theorem~5.3]{Fei98} makes it NP-hard to distinguish unweighted Max-$k$-Cover instances in which $k$ sets cover all $W$ elements of the universe from instances in which every $k$ sets cover at most $(\aconst+\delta/4)W$ elements. These instances are a subclass of weighted coverage. Present the sets in any fixed order and consider the algorithm's final output value $V\in[0,W]$. Since the instance is explicit, every current value query and the final output value can be evaluated in polynomial time. In the YES case, $\E V\ge(\aconst+\delta)W$. In the NO case, feasibility gives $V\le(\aconst+\delta/4)W$ on every outcome.

Fix a rational constant $\tau\in(\aconst+\delta/2,\aconst+3\delta/4)$ and let $N=\lceil8\ln(3)/\delta^2\rceil$. Run the algorithm independently $N$ times, obtaining values $V_1,\ldots,V_N$, and accept if their average exceeds $\tau W$. The NO case is always rejected. In the YES case, the expected average exceeds the threshold by more than $\delta W/4$, so Hoeffding's inequality gives
\[
\Prob\!\left\{\frac1N\sum_{j=1}^{N}V_j\le\tau W\right\}
\le \exp\!\left(-2N(\delta/4)^2\right)
=e^{-N\delta^2/8}\le\frac13.
\]
Since $\delta$, $\tau$, and $N$ are fixed constants, this test runs in polynomial time. Composing it with Feige's polynomial-time reduction from an NP-complete language yields a bounded-error decision algorithm for that language, and hence $\mathsf{NP}\subseteq\mathsf{BPP}$. The same argument applies to an offline algorithm, since only the final feasible output is used.
\end{proof}

\section{Parameters and Finite Estimates for Curvature}\label[appendix]{app:curvature-details}

This appendix gives the accuracy parameters and query bound for \cref{thm:hybrid}, and the finite estimates used in the lower bound of \cref{thm:curvature-main}.

\subsection{Parameters for Modular Preservation}\label[appendix]{app:hybrid-parameters}

\emph{Static accuracy.} Since $f_{\rm cmp}\le f$, the empty output suffices for $\varepsilon\ge1$. For $0<\varepsilon<1$, set
\[
B=\lceil16/\varepsilon\rceil,\qquad
\eta=\delta_{\rm fail}=\varepsilon/64,\qquad \Lambda=\sqrt2.
\]
For a nonempty current set $Y$, let $M=\max_{i\in Y}f(\{i\})$. The mixed potential $\Xi$ satisfies \cref{eq:smooth-bounds}. The routine of \cref{lem:first-order} uses current queries to $g$ and adds the known modular weights exactly. With probability at least $1-\delta_{\rm fail}$, it returns a point with first-order gap at most $\eta M$. On this event, \cref{eq:hybrid-pointwise-core} holds for every feasible comparison set and every fixed disjoint future. The rounding in \cref{lem:pipage} preserves modular value and does not decrease the multilinear value in expectation for each such future. Thus \cref{eq:hybrid-static-core} holds with $\delta_{\rm core}=\eta+\delta_{\rm fail}=\varepsilon/32$.

\emph{Online conversion.} For $k<4B^2$, recompute a mixed-core sample after each arrival with $\kappa=k$ and $R=\varnothing$. Its expected value is at least $\OPT_k(f_{\rm cmp},X_t)-(\varepsilon/32)\OPT_k(f,X_t)$, and successive outputs differ by at most $2k<8B^2$. Using the mixed-core sampler preserves coefficient one on the modular benchmark.

For $k\ge4B^2$, use $L=\lfloor k/B\rfloor$ and $\kappa=k-2L$, with a uniformly chosen replacement window independent of set generation. By \cref{lem:hybrid-checkpoint}, the additive loss is at most $(\varepsilon/32+3/B)\OPT_k(f,X_t)$, where
\[
\frac\varepsilon{32}+\frac3B\le\frac{7\varepsilon}{32}<\varepsilon,
\]
and the worst-case recourse is at most $8B^2+2$. Both branches remain feasible on every outcome. Their recourse bounds also hold when gradient estimates are inaccurate: the small-capacity bound uses only cardinality, and the large-capacity bound follows from the replacement schedule and \cref{lem:lazy}.

\emph{Queries and computation.} There is at most one call to the static routine per arrival. Substituting $\eta=\delta_{\rm fail}=\varepsilon/64$ into \cref{eq:oracle-work}, with $|Y|\le|X_t|$ and $\kappa\le k$, gives the per-arrival bound
\[
O\!\left(|X_t|k^5\varepsilon^{-4}\log\frac{|X_t|k}{\varepsilon}\right).
\]
Rounding uses no value queries, and the remaining operations are polynomial under the computation model in \cref{sec:model-cores}.

\subsection{Singleton Normalization and the Finite Ratio}

\begin{lemma}[Finite curvature lower bound]\label{lem:finite-curvature}
Fix $T\in[\nicefrac{7}{5},\nicefrac{3}{2}]$, integers $m\ge32$ and $k\ge64m$, and a real $\lambda_{\rm mod}\ge0$. Set $n=(m+1)k$, $\vartheta=(1+\lambda_{\rm mod})^{-1}$, and $\mathcal R(T)=w(1)$. Consider any randomized value-oracle algorithm that maintains at most $k$ arrived elements on every outcome and terminates almost surely after each arrival on every instance with total curvature at most $\vartheta$. Assume that, on every such instance with $n+1$ elements and every random outcome, it makes at most $Q$ current value queries before the final arrival and at most $C$ symmetric changes there. Then there exists a $k$-set $A\subseteq X=[n]$, fixed independently of the algorithm's randomness, whose augmented instance $f_{\lambda_{\rm mod}}$ from \cref{eq:modular-hard} has total curvature at most $\vartheta$ and satisfies, for arrival order $1,2,\ldots,n,r$,
\begin{equation}\label{eq:finite-curvature}
\frac{\E f_{\lambda_{\rm mod}}(S_{n+1})}{\OPT_k(f_{\lambda_{\rm mod}},X\cup\{r\})}
\le\frac{\mathcal R(T)+\lambda_{\rm mod}}{1+\lambda_{\rm mod}}
+\frac{130}{m}+\frac{10C+12}{k}
+2(Q+1)e^{-k/(8192m^2)}.
\end{equation}
The expectation is over the algorithm's randomness. Queries and computation after the final arrival are unrestricted.
\end{lemma}

The leading term equals $1-\vartheta+\vartheta\mathcal R(T)$, which is $\rho_\vartheta$ from \cref{thm:curvature-main} at $T=\sqrt2$.

\begin{proof}
The construction and repair estimate are established in \cref{sec:curvature}. We verify the singleton bounds, the reference promise, and the finite ratio.

\emph{Singleton normalization.} A current singleton has normalized counts $(1/k,0)$ or $(0,1/k)$, with imbalances $1/k$ and $-1/(mk)$, respectively. Since $k\ge64m$ and $\delta=1/(64m)$, both lie in the exact flat band. Hence all current singletons have the same value
\[
\sigma_k=\frac{v'(0)}k-\frac{d}{2k^2}
+\frac{32}{m}\bigl(1-e^{-1/k}\bigr),
\qquad 0\le\sigma_k\le\frac5k.
\]
The final singleton is $c_*=c+32\delta=c+1/(2m)$. Both values are independent of $A$, so the added modular term on a current query $S\subseteq X$ is $\lambda_{\rm mod}\sigma_k|S|$, depending only on query size.

Define $\Delta_k=k\sigma_k+c_*$. The identity $v'(0)+c=1$ gives
\[
\Delta_k=1-\frac{d}{2k}+\frac{32k}{m}\bigl(1-e^{-1/k}\bigr)+\frac1{2m}.
\]
Since $0<d\le1$ and $0\le k(1-e^{-1/k})\le1$,
\begin{equation}\label{eq:singleton-balance-finite}
1-\frac1{2k}\le\Delta_k\le1+\frac{65}{2m}\le1+\frac{65}{m}.
\end{equation}
The regularizer and future correction leave an $O(1/m)$ contribution as $k$ grows, so $m$ is chosen for the target accuracy before $k$ is made large.

\emph{A reference instance within the promise.} The algorithm's query and termination assumptions must also apply to the reference execution. With $\psi(s)=v(s)+\phi_m(s)$, define
\begin{equation}\label{eq:curvature-reference}
f_{\mathrm{ref},\lambda_{\rm mod}}(S)
=\psi(|S\cap X|/k)+\lambda_{\rm mod}\sigma_k|S\cap X|
 +(1+\lambda_{\rm mod})c_*\one\{r\in S\}.
\end{equation}
By \cref{app:transcript}, $\psi$ is nonnegative, nondecreasing, and concave, with $\psi(0)=0$. Thus the reference is normalized monotone submodular. Its current singleton values are $(1+\lambda_{\rm mod})\sigma_k$, and for $i\in X\setminus S$,
\[
f_{\mathrm{ref},\lambda_{\rm mod}}(i\mid S)
\ge\lambda_{\rm mod}\sigma_k
=(1-\vartheta)f_{\mathrm{ref},\lambda_{\rm mod}}(\{i\}).
\]
The final element is modular, so the reference has total curvature at most $\vartheta$. Its answers agree with the augmented hard instance on every balanced current query. The reference-execution argument of \cref{lem:transcript} therefore applies within the promised class. Only pre-arrival answers are coupled.

\emph{The finite ratio.} On the good transcript event, \cref{eq:curvature-exact-finite-ratio} bounds the ratio. By \cref{eq:singleton-balance-finite}, its numerator is at most
$\mathcal R(T)+\lambda_{\rm mod}+65(1+\lambda_{\rm mod})/m+5C/k$.
Also $\Delta_k-\sigma_k\ge1-11/(2k)\ge1-6/k$, so its denominator is at least $(1+\lambda_{\rm mod})(1-6/k)$. Set $a=(\mathcal R(T)+\lambda_{\rm mod})/(1+\lambda_{\rm mod})\le1$. Since $k\ge64m\ge12$,
\[
\begin{aligned}
\frac{a+65/m+5C/((1+\lambda_{\rm mod})k)}{1-6/k}
&=a+\frac{6a/k+65/m+5C/((1+\lambda_{\rm mod})k)}{1-6/k}\\
&\le a+\frac{130}{m}+\frac{10C+12}{k}.
\end{aligned}
\]
The last inequality uses $a\le1$, $1/(1-6/k)\le2$, and $\lambda_{\rm mod}\ge0$, so the error bound is uniform in $\lambda_{\rm mod}$.

The good transcript event fails with probability at most $2(Q+1)e^{-k/(8192m^2)}$. On failure, feasibility bounds the ratio by one. Averaging the ratio over a uniformly chosen hidden $k$-set $A$ and the algorithm's random tape yields a fixed $A$ for which \cref{eq:finite-curvature} holds in expectation over the tape alone.
\end{proof}

\section{Further Related Work}\label[appendix]{app:further-related}

\paragraph{Consistent submodular maximization.}
The deterministic model of \citep{DFL+24} has a $1/2$ impossibility bound and algorithms with guarantees approaching $0.3818$. The addition-robust primitive and checkpoint reduction of \citep{DFL+25} allow randomization and give the unrestricted $2/3$ threshold. Under insertions and deletions, \citep{DFL+26} obtains a $(1/2-O(\varepsilon))$ approximation with $O(\varepsilon^{-2})$ worst-case symmetric recourse for cardinality constraints, and a $(1/4-O(\varepsilon))$ approximation with $O(\log k/\varepsilon^2)$ recourse for rank-$k$ matroids. Their framework combines randomized scheduling with routines robust to deletions. The threshold studied here combines polynomially many current value queries with worst-case constant symmetric recourse on insertion-only streams.

\paragraph{Dynamic query complexity.}
Chen and Peng~\citep[Theorem~1.2]{CP22} prove an $\Omega(N/k^3)$ amortized query lower bound for $N$ updates and $k\ge C\log N$, assuming a $0.584$ approximation simultaneously at all times with probability at least $2/3$. The construction uses deletions and permits unrestricted recourse. Our lower bound uses an insertion-only stream, requires only expected approximation at each fixed time, and limits recourse at one final arrival, even with unlimited subsequent queries. Without recourse restrictions, their insertion-only algorithm achieves $1-1/e-\varepsilon$ using $O(\log(k/\varepsilon)/\varepsilon^2)$ amortized queries~\citep[Theorem~1.3]{CP22}.

Other fully dynamic algorithms under cardinality constraints achieve approximations approaching $1/2$ with polylogarithmic amortized query complexity~\citep{LMNF+20} or polynomial dependence on $k$ in the update cost~\citep{Mon20}. We use the corrected full version of \citep{LMNF+20}. For matroid constraints, \citep{DFL+23} gives a $1/4$ approximation with $\widetilde O(k^2)$ expected amortized update time, and \citep{BBG+24} gives query bounds depending on rank and accuracy but independent of the ground-set size. A recent preprint~\citep{BGH+26} maintains a $1/(4+\varepsilon)$ approximation with $\operatorname{poly}(\log n,1/\varepsilon)$ expected amortized oracle queries per update, independent of matroid rank, against an oblivious adversary. Query and update costs are separate from changes to the maintained set.

\paragraph{Earlier occurrences of the constant.}
In the single-pass element-store model, \citep[Theorem~1.2]{HKMY22} proves an $\Omega(n/k^2)$ space lower bound for any fixed improvement over $2-\sqrt2$, even with a strong value oracle on stored elements and the current arrival. Its Lemma~2.1 and Definition~2.3 supply the normalized continuous initial segment used in \cref{sec:lower}, as detailed in Appendix~\ref{app:hkmy-normalization}. At $T=\sqrt2$, this segment meets our supporting-line equality conditions. Access to the full arrived history requires exact query hiding in the modified finite construction.

Randomized composable coresets assign elements randomly to machines and measure the quality of the best feasible set in the union of their summaries. Enlarged \textsc{Greedy} summaries give an expected approximation guarantee approaching $2-\sqrt2$, with a matching limitation for this summary method~\citep[Theorems~4.1 and~4.8]{MZ15}. Efficient \textsc{PseudoGreedy} postprocessing instead gives a guarantee approaching $0.545$~\citep[Theorem~4.9]{MZ15}. The extremal case $r=0$ in \citep[Lemma~4.6, Eq.~(4)]{MZ15} yields, in our notation, $(1-\gamma)^2\le2cd$, with extremal values $c=d=\bconst/2$. Their Theorem~4.8 uses these parameters in a tight example. Our anchored-action LP chooses a distribution over feasible current sets before the legal future is known, and \cref{thm:anchored} proves a separate expected-value guarantee for every fixed legal future from current certificates on a finite grid. The matching lower bound concerns current queries and recourse, combining the hidden-color geometry of \citep{HKMY22}, the final-arrival template of \citep{DFL+25}, and exact hiding from \citep{Von13}.

\paragraph{Adversarial injections.}
The model of \citep{GKRS20} interleaves randomly ordered genuine elements with adversarially placed noise, comparing against the optimum on the genuine elements. Its tree algorithm gives a $0.55$ approximation with storage independent of stream length but exponential in $k$. A state-dependent threshold analysis in \citep[Section~7.4, by Ola Svensson]{WCJ+26} improves the guarantee of the same tree algorithm to $2-\sqrt2$. The benchmark here is instead the optimum over all arrived elements, with an adversarially fixed insertion order.

\paragraph{Streaming memory and access restrictions.}
Sieve-Streaming~\citep{BMKK14} gives a $1/2-\varepsilon$ approximation in one pass with $O(k\log k/\varepsilon)$ memory, and Sieve-Streaming++~\citep{KMZ+19} reduces this to $O(k/\varepsilon)$. The random-order algorithm in \citep{NFTM+18} improves on $1/2$, while its adversarial-order memory lower bound restricts queries to sets of size at most $k$~\citep[Theorem~1.1]{NFTM+18}. The one-way communication bounds in \citep{FNSZ23} constrain message size and allow unrestricted value queries on the players' available elements. The model here imposes no memory limit and permits queries to arbitrary subsets of arrived elements, including sets larger than $k$.

Under a cardinality constraint, \textsc{Swapping}~\citep{CK15} achieves a $1/4$ approximation with at most one exchange per arrival, hence at most two symmetric changes. Preemption~\citep{BFS19} and free disposal~\citep{CHJ+18} permit discarding accepted elements but forbid recalling them afterward. In the consistency model, any previously arrived element remains available for selection, subject to the recourse budget.

\paragraph{Other stability and recourse objectives.}
Consistent clustering measures changes of selected centers: \citep{LV17} introduced this viewpoint for incremental clustering, and \citep{FLNS21} obtained constant-factor $k$-median approximation in general metrics with $\widetilde O(k)$ total center swaps. For fully dynamic $k$-center, \citep{LHG+24} gives worst-case constant recourse, while \citep{BCG+24} studies approximation, update time, and recourse jointly. Dynamic submodular cover~\citep{GL20} bounds changes in a covering solution as the objective changes, and dynamic discrepancy minimization~\citep{GGK+22} counts changes of signs or edge directions. These problems impose different feasibility constraints or movement objectives.

Consistent online optimization~\citep{KJKLV19} trades cumulative regret against the number of rounds in which the action changes. Submodular chasing~\citep{BNW25} measures total expected movement against trajectories satisfying prescribed value targets. Dynamic contention resolution~\citep{FKLS26} rounds changing fractional solutions with total expected recourse controlled by their total $\ell_1$ movement. Its dynamic submodular applications compare total movement against a minimum-recourse benchmark on the same update sequence. Our recourse bound holds separately at every arrival and for every outcome of the random choices, while approximation is in expectation at each fixed time.

\paragraph{Correlation gaps and prices.}
Correlation-gap inequalities compare independent sampling with correlated distributions of the same marginals for a fixed objective~\citep{ADSY10}. Non-oblivious potentials~\citep{FW14} and weighted gradient integrals~\citep{ZDC+22} also support approximation guarantees in submodular maximization. In \cref{sec:scale}, a potential determined by current values supplies one price vector for comparisons with every fixed legal future. Both the point and the prices are computed before that future is known. \Cref{sec:curvature} combines the price inequality with an exact modular identity to retain coefficient one on the modular component.

\section*{AI Methodology}
We used ChatGPT (OpenAI, including GPT-6) to assist with exploring proof strategies and lower-bound constructions, checking technical arguments, identifying relevant literature, and drafting and revising parts of the manuscript. In particular, we used the tool to explore and refine the lower-bound construction in Section 5 and Appendix B, as well as to review proofs and improve exposition in Sections 4 and 6–8 and the corresponding appendices. The authors critically evaluated AI-assisted suggestions, selected and refined the final mathematical arguments, and verified the correctness and originality of the results, including references. The authors take full responsibility for all content.

\renewcommand{\bibfont}{\normalsize}
\setlength{\bibsep}{4pt plus 1pt minus 1pt}
\bibliographystyle{alpha}
\bibliography{references}
\end{document}